\documentclass[11pt,english]{article}
\usepackage{lmodern}
\usepackage[T1]{fontenc}
\usepackage{geometry}
\usepackage{babel}
\usepackage{array}
\usepackage{float}
\usepackage{units}
\usepackage{mathtools}
\usepackage{multirow}
\usepackage{amsmath}
\usepackage{amsthm}
\usepackage{amssymb}
\usepackage[unicode=true]
 {hyperref}

\makeatletter

\providecommand{\tabularnewline}{\\}

\numberwithin{equation}{section}
\numberwithin{figure}{section}
\theoremstyle{plain}
\newtheorem{thm}{\protect\theoremname}[section]
\theoremstyle{remark}
\newtheorem{rem}[thm]{\protect\remarkname}
\theoremstyle{definition}
\newtheorem{defn}[thm]{\protect\definitionname}
\theoremstyle{plain}
\newtheorem{lem}[thm]{\protect\lemmaname}
\theoremstyle{plain}
\newtheorem{cor}[thm]{\protect\corollaryname}
\theoremstyle{remark}
\newtheorem{claim}[thm]{\protect\claimname}

\usepackage{amsthm}
\usepackage{amsbsy}
\usepackage{amsmath}
\usepackage{amsfonts}
\usepackage{bm}
\usepackage{hyperref}
\usepackage{url}

\DeclareMathOperator{\mmc}{mmc}
\DeclareMathOperator{\poly}{poly}
\DeclareMathOperator{\polylog}{polylog}

\usepackage{float}
\usepackage{xspace}
\usepackage[ruled,linesnumbered,vlined]{algorithm2e}

\let\oldnl\nl
\newcommand{\nonl}{\renewcommand{\nl}{\let\nl\oldnl}}

\makeatother

\providecommand{\claimname}{Claim}
\providecommand{\corollaryname}{Corollary}
\providecommand{\definitionname}{Definition}
\providecommand{\lemmaname}{Lemma}
\providecommand{\remarkname}{Remark}
\providecommand{\theoremname}{Theorem}

\begin{document}
\global\long\def\N{\mathbb{N}}%
\global\long\def\R{\mathbb{R}}%
\global\long\def\uni{\cup}%
\global\long\def\G{G=\left(V,E\right)}%
\global\long\def\Gw{G=\left(V,E,w\right)}%
\global\long\def\eps{\mathbb{{\cal \varepsilon}}}%
\global\long\def\delt{\mathbb{{\cal \delta}}}%
\global\long\def\sseq{\subseteq}%
\global\long\def\O#1{O\text{\negmedspace}\left(#1\right)}%
\global\long\def\letdef{\leftarrow}%
\global\long\def\Ot#1{\widetilde{O}\text{\negmedspace}\left(#1\right)}%
\global\long\def\Oe#1{O_{\eps}\text{\negmedspace}\left(#1\right)}%
\global\long\def\Oet#1{\widetilde{O}_{\eps}\text{\negmedspace}\left(#1\right)}%
\global\long\def\o#1{o(#1)}%
\global\long\def\D{\mathcal{D}}%
\global\long\def\S{\mathcal{S}}%
\global\long\def\A{\mathcal{A}}%
\global\long\def\lsw{\Lambda^{\negthinspace sw}}%
\makeatletter
\def \@fnsymbol#1{\ensuremath{\ifcase#1 \or \dagger \or \ddagger \or * \or \textsection \else\@ctrerr\fi}}
\makeatother
\title{Almost-Smooth Histograms and Sliding-Window Graph Algorithms\thanks{A preprint of the paper is available at \protect\href{http://arxiv.org/abs/1904.07957}{arXiv:1904.07957}}}
\author{Robert Krauthgamer\thanks{Weizmann Institute of Science. Work partially supported by ONR Award
N00014-18-1-2364, the Israel Science Foundation grant \#1086/18, and
a Minerva Foundation grant. Part of this work was done while the author
was visiting the Simons Institute for the Theory of Computing. Email:
$\texttt{robert.krauthgamer@weizmann.ac.il}$, $\texttt{david.reitblat@gmail.com}$} \and David Reitblat$^{\ddagger}$\thanks{Currently at Playtika AI Research Lab. Email: $\texttt{davidre@playtika.com}$}}
\maketitle
\begin{abstract}
We study algorithms for the sliding-window model, an important variant
of the data-stream model, in which the goal is to compute some function
of a fixed-length suffix of the stream. We extend the $\emph{smooth-histogram}$
framework of Braverman and Ostrovsky (FOCS 2007) to almost-smooth
functions, which includes all subadditive functions. Specifically,
we show that if a subadditive function can be $\left(1+\eps\right)$-approximated
in the $\emph{insertion-only}$ streaming model, then it can be $\left(2+\eps\right)$-approximated
also in the $\emph{sliding-window}$ model with space complexity larger
by factor $\O{\eps^{-1}\log w}$, where $w$ is the window size.

We demonstrate how our framework yields new approximation algorithms
with relatively little effort for a variety of problems that do not
admit the smooth-histogram technique. For example, in the frequency-vector
model, a symmetric norm is subadditive and thus we obtain a sliding-window
$\left(2+\eps\right)$-approximation algorithm for it. Another example
is for streaming matrices, where we derive a new sliding-window $\left(\sqrt{2}+\eps\right)$-approximation
algorithm for Schatten $4$-norm. We then consider graph streams and
show that many graph problems are subadditive, including maximum submodular
matching, minimum vertex-cover, and maximum $k$-cover, thereby deriving
sliding-window $\O 1$-approximation algorithms for them almost for
free (using known insertion-only algorithms). Finally, we design for
every $d\in\left(1,2\right]$ an artificial function, based on the
maximum-matching size, whose almost-smoothness parameter is exactly
$d$.
\end{abstract}
\newpage{}

\section{Introduction}

Nowadays, there is a growing need for algorithms to process huge data
sets. The Internet, including social networks and electronic commerce,
as well as astronomical and biological data, provide new challenges
for computer scientists and mathematicians, since traditional algorithms
are not able to handle such massive data sets in a reasonable time.
First, the data is too big to be stored on a single machine. Second,
even algorithms with time complexity $\O{n^{2}}$ could be too slow
in practice. Third, and most important, the data could change over
time, and algorithms should cope with these dynamic changes. Therefore,
several models of computation over Big Data are studied, such as parallel,
distributed, and streaming algorithms.

We concentrate on the $\emph{streaming model}$ (see e.g. \cite{Muthukrishnan:2005:DSA:1166409.1166410,Babcock:2002:MID:543613.543615,Aggarwalbook}),
where the data is given as a sequence of items (or updates) in some
order (usually adversarial), and the algorithm can read the data only
in that order. Often, the algorithm can only read the data once, although
there are also algorithms for multiple passes. More concretely, a
$\emph{stream}$ is a (possibly infinite) sequence $\S=\left\langle \sigma_{1},\sigma_{2},\ldots,\sigma_{i},\ldots\right\rangle $,
where each item $\sigma_{i}$ belongs to some universe $U$. The length
of the stream, as well as the size of $U$, is assumed to be huge,
such that storing the entire stream, or even a constant-size information
for each item in $\S$, is impractical. A streaming algorithm $\A$
takes $\S$ as input and computes some function $f$ of the stream
$\S$. This means that the algorithm has access to the input in a
$\emph{streaming fashion}$, i.e., $\A$ can read the input once and
only in the order it is given and at every time $t$ the algorithm
may be asked to evaluate $f$ on the prefix $\S_{t}=\left\langle \sigma_{1},\ldots,\sigma_{t}\right\rangle $,
called a $\emph{query}$ at time $t$ for $f\left(\S_{t}\right)$.
We only consider here the $\emph{insertion-only}$ model, where all
updates are positive, i.e., only adding items to the underlying structure
(in some other models, the deletion of previously added items is also
allowed).

In many streaming scenarios, computing the exact value of $f$ is
computationally prohibitive or even impossible. Hence, the goal is
to design a streaming algorithm whose output approximates $f$ (often
with high probability). As usual, it should have low space complexity,
update time, and query time, see Remark \ref{remark_1.1}.

The $\emph{sliding-window}$ model, introduced by Datar, Gionis, Indyk
and Motwani \cite{StreamStatistics}, has become a popular model for
processing (infinite) data streams, where older data items should
be ignored, as they are considered obsolete. In this model, the goal
is to compute a function $f$ on a suffix of the stream, referred
to as the $\emph{active window}$ $W$. Throughout, the size $w$
of the active window $W$ is assumed to be known (to the algorithm)
in advance. At a point in time $t$, we denote the active window by
$W_{t}=\left\langle \sigma_{t-w+1},\ldots,\sigma_{t}\right\rangle $,
or $W$ for short when $t$ is clear from the context. The goal is
to approximate $f\left(W_{t}\right)$, and possibly provide a corresponding
object, e.g., a feasible matching in a graph when the stream is a
sequence of edges and $f$ is the maximum-matching size. For a randomized
algorithm, we require that a single query at any time $t$ succeeds
with probability at least $1-\delt$. 

Datar et al. \cite{StreamStatistics} noted that in the sliding-window
model there is a lower bound of $\Omega\left(w\right)$ if deletions
are allowed, even for relatively simple tasks like approximating (within
factor $2$) the number of distinct items in a stream. Therefore,
we assume throughout that the stream $\S$ has only insertions, and
no deletions.

A widely studied streaming model is the $\emph{graph-streaming}$
model (see e.g. \cite{Feigenbaum:2005:GPS:1132633.1132638,McGregor:2014:GSA:2627692.2627694}),
where the stream $\S$ consists of a sequence of edges (possibly with
some auxiliary information, like weights) of an underlying graph $\G$.\footnote{All our definitions, e.g., Definitions \ref{def:matching} and \ref{def:VC},
as well as Corollary \ref{cor:MM=000026VC2AS}, extend naturally to
hypergraphs.} We assume that $V=\left[n\right]$ for a known value $n\in\N$ and
$G$ is a simple graph without parallel edges. This graph-streaming
model is sometimes studied in the $\emph{semi-streaming}$ model,
where algorithms are allowed to use $\O{n\cdot\polylog\left(n\right)}$
space. Observe that for dense graphs, an algorithm in this model cannot
store the whole graph, but it can store $\polylog\left(n\right)$
information for each vertex. We slightly abuse the notation of a graph
function $f\left(G\right)$ and extend it to a stream of edges $\S$
using the convention $f\left(\S\right)\coloneqq f\left(G\right)$
where $G$ is the graph defined by the edges in the stream $\S$.
\begin{rem}
\label{remark_1.1}Throughout, $\emph{space complexity}$ refers to
the storage requirement of an algorithm during the entire input stream,
measured in bits. $\emph{Update time}$ refers to the time complexity
of processing a single update from the stream in the RAM model. $\emph{Query time}$
refers to the time complexity of reporting an output at a single point
in time.
\end{rem}

Crouch, McGregor and Stubbs \cite{Crouch2013} initiated the study
of graph problems in the sliding-window model, and designed algorithms
for several basic graph problems, such as $k$-connectivity, bipartiteness,
sparsification, minimum spanning tree and spanners. They also showed
approximation algorithms for maximum-matching and for maximum-weight
matching. We shall focus on two well-known and closely related optimization
problems, maximum-matching and minimum vertex-cover, defined below.

\subsection{Basic Terminology}
\begin{defn}
\label{def:matching}A $\emph{matching}$ in a graph $\G$ is a set
of edges $M\sseq E$ that are disjoint, i.e., no two edges have a
common vertex. Denote by $m\left(G\right)$ the maximum size of a
matching in $G$. A matching of maximal size (number of edges) is
called a $\emph{maximum-cardinality matching}$, and is usually referred
to as a $\emph{maximum-matching}$. In an edge-weighted graph $G$,
a $\emph{maximum-}$ $\emph{weight matching}$ is a matching with
maximal sum of weights.
\end{defn}

\begin{defn}
\label{def:VC}A subset $C\sseq V$ of the vertices of a graph $\G$
is called a $\emph{vertex-cover}$ of $G$ if each edge $e\in E$
is incident to at least one vertex in $C$. Denote by $VC\left(G\right)$
the smallest size of a vertex-cover of $G$.
\end{defn}

We will use the terminology of Feige and Jozeph \cite{Feige:2015:SEA:2688073.2688101}
to distinguish between estimation and approximation of optimization
problems (where the goal is to find a feasible solution of optimal
value). An $\emph{approximation algorithm}$ is required to output
a feasible solution whose value is close to the value of an optimal
solution, e.g., output a feasible matching of near-optimal size. An
$\emph{estimation algorithm}$ is required to only output a value
close to that of an optimal solution, without necessarily outputting
a corresponding feasible solution, e.g., estimate the size of a maximum-matching,
without providing a corresponding matching.
\begin{defn}
\label{def:Notation}The notation $\Ot s$ hides polylogarithmic dependence
on $s$, i.e., $\Ot s=\O{s\cdot\polylog\left(s\right)}$. To suppress
dependence on $\eps$ we write $\Oe s=\O{s\cdot f\left(\eps\right)}$,
where $f:\R^{+}\rightarrow\R^{+}$ is some positive function.\footnote{Throughout, every dependence on $\eps$ is polynomial, i.e., in our
case $\Oe s=\O{s\cdot\poly\left(\eps^{-1}\right)}$.} We also combine both notations and define $\Oet s=\Oe{s\cdot\polylog\left(s\right)}$.
\end{defn}

\begin{defn}
\label{def:rand_alg}For $\delt\in\left[0,1\right)$ and $C\ge1$,
a randomized algorithm $\Lambda$ is said to $\left(C,\delt\right)$-approximate
a function $f$ if on every input stream $\S$, its output $\Lambda\left(\S\right)$,
upon a query (at a single point of time), satisfies
\[
\Pr\left[f\left(\S\right)\le\Lambda\left(\S\right)\le C\cdot f\left(\S\right)\right]\ge1-\delt.
\]
If $\Lambda$ is a deterministic algorithm then $\delt=0$, and we
say in short that it $C$-approximates $f$. We sometime use this
shorter terminology and omit $\delt$ also for randomized approximation
algorithms, when $\delt$ is a fixed constant, say $0.1$. It is well-known
that every algorithm (that reports a real value) with success probability
$0.9$ can be amplified to success probability $1-\frac{1}{\poly\left(n\right)}$
using $\O{\log\left(n\right)}$ independent repetitions and reporting
the median.
\end{defn}

\begin{rem}
\label{rem:rand_alg}Usually we use Definition \ref{def:rand_alg}
with an extra approximation factor $1\pm\eps$, for $\eps\in\left(0,\frac{1}{2}\right)$,
in the following manner. If algorithm $\Lambda$ satisfies
\[
\Pr\left[\left(1-\eps\right)f\left(\S\right)\le\Lambda\left(\S\right)\le\left(1+\eps\right)C\cdot f\left(\S\right)\right]\ge1-\delt,
\]
then algorithm $\Lambda^{\prime}\coloneqq\frac{1}{1-\eps}\Lambda$
satisfies Definition \ref{def:rand_alg} with $C^{\prime}\coloneqq\frac{\left(1+\eps\right)}{\left(1-\eps\right)}C<\left(1+4\eps\right)C$.
\end{rem}

\subsection{Our Contribution}

We introduce an adaptation of the smooth-histogram technique of Braverman
and Ostrovsky \cite{Braverman:2007:SHS:1333875.1334202} to a more
general family of functions, that we call almost-smooth, and demonstrate
that it is applicable in a variety of settings, including frequency-vector
streams, graph streams, and matrix streams. In fact, many of our examples
follow a single reasoning --- these functions are subadditive (defined
below) and thus $2$-almost-smooth --- and some of them (e.g., symmetric
norms and maximum-matching) are not smooth and thus do not admit the
more restricted smooth-histogram technique. Furthermore, we show artificial
examples where the almost-smoothness parameter can take any value
in the range $\left(1,2\right]$.

Similarly to \cite{Braverman:2007:SHS:1333875.1334202}, the main
idea in our framework is to maintain several instances of an insertion-only
algorithm on different suffixes of the stream, such that at every
point in time, the active window $W$ and the largest suffix of it
that is maintained have similar value of $f$. We show that for an
almost-smooth $f$ our overall algorithm achieves a good approximation
of $f\left(W\right)$. Using this technique we obtain new sliding-window
algorithms for several problems that admit an insertion-only streaming
algorithm like estimating a symmetric norm of a frequency-vector or
submodular matching in a graph stream. For more details on the smooth-histogram
framework of \cite{Braverman:2007:SHS:1333875.1334202} see Appendix
\ref{appendix:Smooth-Histogram-Framework}.

\paragraph{Almost-Smooth Functions}

We start with an overview of our notion of almost-smooth functions;
for a more formal treatment see Definition \ref{def:ASfunction} and
Remark \ref{rem:ASfunctionRem}. A function $f$ defined on streams
is said to be $\emph{left-monotone}$ (non-decreasing) if for every
two disjoint segments $A,B$ of a stream $f\left(B\right)\le f\left(AB\right)$,
where $AB$ denotes their concatenation. Informally, we say that a
left-monotone function $f$ is $\emph{ d-almost-smooth}$ if $\frac{f\left(B\right)}{f\left(AB\right)}\le d\cdot\frac{f\left(BC\right)}{f\left(ABC\right)}$
for all disjoint segments $A,B,C$; this means that whenever $f\left(B\right)$
approximates $f\left(AB\right)$ within some factor, appending any
segment $C$ will maintain this approximation up to an extra factor
$d$. For example, the maximum-matching size is $2$-almost-smooth
(see Corollary \ref{cor:MM=000026VC2AS}), which means that if $A,B,C$
are disjoint sets of edges and $m\left(B\right)$ is a $\left(1+\eps\right)$-approximation
of $m\left(AB\right)$, then for every sequence $C$ of additional
edges, $m\left(BC\right)$ would $\left(\left(1+\eps\right)2\right)$-approximate
$m\left(ABC\right)$.

\paragraph{Algorithms for Almost-Smooth Functions}

For almost-smooth functions we show a general transformation of an
approximation algorithm in the insertion-only model to the sliding-window
model.
\begin{thm}
[Informal version of Theorem \ref{thm:Almost-Smooth}] \label{thm:Almost-Smooth-1}
Suppose function $f$ is $d$-almost-smooth and can be $C$-approximated
by an insertion-only algorithm $\Lambda$. Then there exists a sliding-window
algorithm $\lsw$ that $\left(\left(1+\eps\right)dC^{2}\right)$-approximates
$f$, with only a factor $\O{\eps^{-1}\log w}$ larger space and update
time than $\Lambda$.
\end{thm}

We show in Lemma \ref{lem:subadditive2AS} that every subadditive
function (defined below) $f$ is $2$-almost-smooth; it follows (using
Theorem \ref{thm:Almost-Smooth-1}) that every insertion-only algorithm
for approximating such $f$ can be adapted to the sliding-window model
with a small overhead in the approximation ratio, space complexity,
and update time.
\begin{defn}
\label{def:subadditive}A function $f$ is said to be $\emph{subadditive}$
if for every disjoint segments $A,B$ of a stream it holds that $f\left(AB\right)\le f\left(A\right)+f\left(B\right)$.
\end{defn}

\paragraph{Frequency-Vector Streams}

In the frequency-vector model, where the stream consists of additive
updates to an underlying vector $x\in\R^{n}$, we show that every
symmetric norm is subadditive and thus admits our framework of almost-smoothness.
We can then use a result of B$\l$asiok et al. \cite{BBC17} that
provides a $\left(1+\eps\right)$-approximation randomized streaming
algorithm for every symmetric norm, to deduce the following theorem.
\begin{thm}
[Informal version of Theorem \ref{thm:symmetric_norm}] \label{thm:symmetric_norm-1}
Every symmetric norm on $\R^{n}$ admits a $\left(2+\eps\right)$-approximation
sliding-window randomized algorithm with space complexity $\Oet L$,
where $L$ is a certain quantity associated with the norm.
\end{thm}

Previously, sliding-window algorithms were known only for $\ell_{p}$
norms \cite{Braverman:2007:SHS:1333875.1334202,WZ21}. We also exemplify
a norm that does not admit the more restricted smooth-histogram technique
of Braverman and Ostrovsky. We prove that the top-$k$ norm for $k=\frac{n}{2}$
is not $d$-almost-smooth for any $d<2$, although it is subadditive
and thus $2$-almost-smooth. Moreover, this norm has $L=\O{\polylog\left(n\right)}$,
where $L$ is the associated quantity from Theorem \ref{thm:symmetric_norm-1},
and therefore our algorithm has space complexity $\Oe{\polylog\left(nw\right)}$
(see Section \ref{subsec:Symmetric-Norms}). Note that we also use
this algorithm to obtain an $\O 1$-approximation sliding-window algorithm
for Max $k$-Cover (see Section \ref{subsec:Max-k-Cover}).

\paragraph{Matrix Streams}

For matrix streams, where the stream consists of $n$ row vectors
(in $\R^{m}$) that form an $n\times m$ matrix, we show that the
Schatten $p$-norm, for $p\ge2$, is $\sqrt{2}$-almost-smooth. For
$p=4$ we use an insertion-only algorithm of Braverman et al. \cite{Braverman2020SchattenNI}
to obtain a sliding-window algorithm.
\begin{thm}
[Informal version of Corollary \ref{cor:schatten_4-norm_alg}] \label{thm:schatten_4-norm_alg-1}
There exists a sliding-window algorithm that $\left(\sqrt{2}+\eps\right)$-approximates
the Schatten $4$-norm of a matrix using $\Oe{\polylog\left(nw\right)}$
bits of space.
\end{thm}

\paragraph{Graph Streams}

In the graph-streams domain, where the stream consists of edges that
form an underlying graph, we consider several problems, starting in
Section \ref{sec:AS-example} with two examples of graph problems
suitable for our framework. The first one is maximum submodular matching,
which is a generalization of the maximum-weight matching problem to
submodular objective functions. The second one is Max $k$-Cover,
which is dual to Vertex-Cover and whose goal is to cover as many edges
as possible using only $k$ vertices. We prove that both problems
are subadditive and hence $2$-almost-smooth. For both problems we
present an $\O 1$-approximation sliding-window algorithms, with space
complexity $\Ot n$, see Corollaries \ref{cor:MSM_alg} and \ref{cor:max-k-cover_apprx_alg},
respectively.

Recent research on graph streams \cite{Esfandiari:2018:SAE:3266298.3230819,mcgregor_et_al:LIPIcs:2016:6640,10.1007/978-3-662-48350-3_23,Chitnis:2016:KVS:2884435.2884527}
addressed maximum-matching size in a restricted family of graphs.
Specifically, McGregor and Vorotnikova \cite{mcgregor_et_al:OASIcs:2018:8295},
improving over Cormode et al. \cite{cormode_et_al:LIPIcs:2017:7849},
designed a $\polylog\left(n\right)$-space algorithm for estimating
the maximum-matching size in arboricity-$\alpha$ graphs within factor
$\O{\alpha}$. Recall that the $\emph{arboricity}$ of a graph $\G$
is the minimal $\alpha\ge1$ such that the set of edges $E$ can be
partitioned into at most $\alpha$ forests. For example, it is well
known that every planar graph has arboricity $\alpha\le3$, see e.g.
\cite{GROSSI1998121}.

Using our generalization of the smooth-histogram technique we provide
several algorithms for estimating maximum matching and minimum vertex-cover
in bounded-arboricity graphs. In particular, we show the following
theorem for maximum matching in Section \ref{sec:Maximum-Matching}.
We compare it in Table \ref{tab:SWMM} with the known insertion-only
algorithms \cite{cormode_et_al:LIPIcs:2017:7849,mcgregor_et_al:OASIcs:2018:8295}.
Note that a straightforward application of Theorem \ref{thm:Almost-Smooth-1}
leads to a weaker result, where the dependency on the arboricity $\alpha$
is quadratic.
\begin{thm}
\label{thm:MM_SW} For every $\eps,\delt\in\left(0,\frac{1}{2}\right)$,
there is a sliding-window $\left(\left(2+\eps\right)\left(\alpha+2\right),\delt\right)$-estimation
algorithm for maximum-matching in graphs of arboricity $\alpha$,
with space complexity $\O{\eps^{-3}\log^{4}n\log\tfrac{1}{\eps\delt}}$
bits and update time $\O{\eps^{-3}\log^{3}n\log\tfrac{1}{\eps\delt}}$.
\end{thm}

\begin{table}[H]
\centering{}%
\begin{tabular}{|>{\centering}m{3cm}ccc|}
\hline 
Stream & Approx. & Space & Reference\tabularnewline
\hline 
\hline 
\vspace{0.05cm}
insertion-only\vspace{0.05cm}
 & $22.5\alpha+6$ & $\O{\alpha\log^{2}n}$ & \cite{cormode_et_al:LIPIcs:2017:7849}\tabularnewline
\hline 
insertion-only & $\left(\alpha+2\right)+\eps$ & $\Oe{\log^{2}n}$ & \cite{mcgregor_et_al:OASIcs:2018:8295}\tabularnewline
\hline 
sliding-window & $\left(2+\eps\right)\left(\alpha+2\right)$ & $\Oe{\log^{4}n}$ & Theorem \ref{thm:MM_SW}\tabularnewline
\hline 
\end{tabular}\caption{\label{tab:SWMM}\negmedspace{}Randomized estimation algorithms for
maximum-matching in graphs of arboricity $\alpha$.}
\end{table}

We design several algorithms also for (estimation and approximation
of) vertex-cover based on its relation to maximum matching as summarized
in Table \ref{tab:SWMVC}. First, for general graphs, our sliding-window
algorithm (Theorem \ref{thm:SW_VC}) improves the previous approximation
ratio, essentially from 8 to 4, using the same space complexity. The
improvement comes from utilizing the almost-smoothness of the greedy
matching (instead of the optimal vertex-cover). Next, for VDP (vertex-disjoint
paths\footnote{A graph $\G$ is said to be VDP if $G$ is a union of vertex-disjoint
paths. This family was used to prove lower bounds in \cite{Esfandiari:2018:SAE:3266298.3230819}.}) and forest graphs (arboricity $\alpha=1$) we compare our two sliding-window
estimation algorithms to one another, as well as to the known turnstile
estimation algorithm \cite{Otniel16}. Notice that Theorem \ref{thm:SW-forVC-polylog}
applies also to VDP graphs (as a special case of forests) and thus
offers a much better space complexity than Theorem \ref{thm:SW-forVC-sqrt},
$\Oe{\log^{4}n}$ compared to $\Oet{\sqrt{n}}$, although the approximation
ratio is slightly worse.

\begin{table}[H]
\centering{}%
\begin{tabular}{|>{\raggedright}m{2.85cm}|>{\centering}p{1.3cm}|>{\centering}m{3cm}ccc|}
\hline 
Problem & Graphs & Stream & Approx. & Space & Reference\tabularnewline
\hline 
\hline &  & \vspace{0.05cm}
insertion-only\vspace{0.05cm}
 & $2$ & $\O{n\log n}$ & Folklore\tabularnewline
\cline{3-6} \cline{4-6} \cline{5-6} \cline{6-6} 
\centering vertex-cover

(approximation) & general & \vspace{0.05cm}
sliding-window\vspace{0.05cm}
 & $8+\eps$ & $\Oe{n\log^{2}n}$ & \cite{Otniel16}\tabularnewline
\cline{3-6} \cline{4-6} \cline{5-6} \cline{6-6} 
 &  & \vspace{0.05cm}
sliding-window\vspace{0.05cm}
 & $4+\eps$ & $\Oe{n\log^{2}n}$ & Theorem \ref{thm:SW_VC}\tabularnewline
\hline 
\hline 
\multirow{5}{2.85cm}{\centering  vertex-cover size (estimation)} & VDP & \vspace{0.05cm}
insertion-only\vspace{0.05cm}
 & $1.5-\eps$ & $\Omega\left(\sqrt{n}\right)$ & \cite{Esfandiari:2018:SAE:3266298.3230819}\tabularnewline
\cline{2-6} \cline{3-6} \cline{4-6} \cline{5-6} \cline{6-6} 
 & VDP & \vspace{0.05cm}
turnstile\vspace{0.05cm}
 & $1.25+\eps$ & $\Oe{\sqrt{n}\log^{2}n}$ & \cite{Otniel16}\tabularnewline
\cline{2-6} \cline{3-6} \cline{4-6} \cline{5-6} \cline{6-6} 
 & VDP & \vspace{0.05cm}
sliding-window\vspace{0.05cm}
 & $3.125+\eps$ & $\Oe{\sqrt{n}\log^{4}n}$ & Theorem \ref{thm:SW-forVC-sqrt}\tabularnewline
\cline{2-6} \cline{3-6} \cline{4-6} \cline{5-6} \cline{6-6} 
 & forests & \vspace{0.05cm}
insertion-only\vspace{0.05cm}
 & $2+\eps$ & $\Oe{\log^{2}n}$ & \cite{mcgregor_et_al:OASIcs:2018:8295}\tabularnewline
\cline{2-6} \cline{3-6} \cline{4-6} \cline{5-6} \cline{6-6} 
 & forests & \vspace{0.05cm}
sliding-window\vspace{0.05cm}
 & $4+\eps$ & $\Oe{\log^{4}n}$ & Theorem \ref{thm:SW-forVC-polylog}\tabularnewline
\hline 
\end{tabular}\caption{\label{tab:SWMVC}Randomized streaming algorithms for vertex-cover
in different settings. The results for vertex-cover size in forests
(including VDP graphs) apply also to maximum-matching size, since
the two quantities are equivalent by K\H{o}nig's Theorem, see Remarks
\ref{rem:2_app_of_E^*} and \ref{rem:VCforForests}.}
\end{table}

\section{\label{sec:SW-alg}Sliding-Window Algorithm for Almost-Smooth Functions}

In this section we generalize the smooth-histogram framework of Braverman
and Ostrovsky \cite{Braverman:2007:SHS:1333875.1334202} to functions
that are almost smooth, as per our new definition, and show that the
family of subadditive functions are almost smooth. We show that several
graph problems satisfy the subadditivity property, e.g., the maximum-matching
size and the minimum vertex-cover size. In the next two sections we
use these results to design sliding-window algorithms for those graph
problems.

\subsection{Almost-Smooth Functions}

Recall that for disjoint segments $A,B$ of a stream, we denote by
$AB$ their concatenation. We use the parameter $n$ to denote some
measure of a stream which will be clear from the context. For example,
for graph streams $n$ is the number of vertices in the underlying
graph. We extend the definition of smoothness due to \cite{Braverman:2007:SHS:1333875.1334202}
as follows.

\begin{defn}
\textbf{\label{def:ASfunction} (Almost-Smooth Function)} A real-valued
function $f$ defined on streams is $\emph{\ensuremath{\left(c,d\right)}-almost-smooth}$,
for $c,d\ge1$, if it has the following properties:
\begin{enumerate}
\item Non-negative: for every stream $A$ it holds that $f\left(A\right)\ge0$.
\item $c$-left-monotone: for every disjoint segments $A,B$ of a stream
it holds that $f\left(B\right)\le c\cdot f\left(AB\right)$.
\item Bounded: for every stream $A$ it holds that $f\left(A\right)\le\poly\left(n\right)$.
\item Almost smooth: for every disjoint segments $A,B,C$ of the stream,
\[
\frac{f\left(B\right)}{f\left(AB\right)}\le d\cdot\frac{f\left(BC\right)}{f\left(ABC\right)}
\]
whenever $f\left(AB\right)\neq0$ and $f\left(ABC\right)\neq0$.
\end{enumerate}
\end{defn}

~
\begin{rem}
\label{rem:ASfunctionRem}Almost-smoothness means that appending any
segment $C$ at the end of the stream preserves the approximation
of $f\left(B\right)$ by $f\left(AB\right)$, up to a multiplicative
factor $d$. Observe that property $4$ is equivalent to the condition
that for every $\eps>0$ and every disjoint segments of the stream
$A,B$ and $C$, 
\[
\begin{array}{ccccc}
\eps\cdot f\left(AB\right)\le f\left(B\right) &  & \implies &  & \eps\cdot f\left(ABC\right)\le d\cdot f\left(BC\right).\end{array}
\]
Obviously, $\eps>c$ is vacuous by property 2, hence it suffices to
consider $0<\eps\le c$. Throughout, it is more convenient to use
this equivalent condition.
\end{rem}

For generality we defined $\left(c,d\right)$-almost-smooth for any
$c\ge1$, but in our applications $c=1$, in which case we simply
omit $c$ and refer to such functions as $d$-almost-smooth.
\begin{rem}
\label{rem:(a,b)-smooth}In the original definition of smoothness
from \cite{Braverman:2007:SHS:1333875.1334202}, $c=d=1$, and property
$4$ is stated as follows (after some simplification). A function
$f$ is $\left(\eps,\beta\left(\eps\right)\right)$-smooth if for
every $\eps\in\left(0,1\right)$ there exists $\beta=\beta\left(\eps\right)$
such that $\beta\le\eps$ and
\[
\begin{array}{ccccc}
\left(1-\beta\left(\eps\right)\right)\cdot f\left(AB\right)\le f\left(B\right) &  & \implies &  & \left(1-\eps\right)\cdot f\left(ABC\right)\le f\left(BC\right).\end{array}
\]
Observe that this definition implies $d$-almost-smoothness if $d\coloneqq\sup\limits _{0<\eps<1}\frac{1-\beta\left(\eps\right)}{1-\eps}$
is bounded. In most applications, it suffices to consider $0<\eps\le\frac{1}{2}$,
and then $\frac{1-\beta\left(\eps\right)}{1-\eps}\le2$ is bounded. 
\end{rem}

We say that a function $f$ is $\emph{monotone}$ (non-decreasing)
if it is left-monotone and right-monotone, i.e., for every disjoint
segments $A,B$ of a stream $f\left(AB\right)\ge f\left(B\right)$
and $f\left(AB\right)\ge f\left(A\right)$.
\begin{lem}
\label{lem:subadditive2AS}Every subadditive, non-negative, bounded
and monotone function $f$ is $2$-almost-smooth.
\end{lem}

\begin{proof}
The first three requirements are clear, as $f$ is assumed to be non-negative,
bounded and monotone. Hence, we are only left to show the almost-smoothness
property. Let $\eps\in\left(0,1\right]$ and let $A,B$ and $C$ be
disjoint segments of the stream satisfying $\eps f\left(AB\right)\le f\left(B\right)$.
Observe that $f\left(AB\right)+f\left(BC\right)\ge f\left(AB\right)+f\left(C\right)\ge f\left(ABC\right)$,
because $f$ is subadditive and monotone, and therefore,
\[
\begin{alignedat}{1}2f\left(BC\right)\ge & f\left(B\right)+f\left(BC\right)\ge\eps\cdot f\left(AB\right)+f\left(BC\right)\\
\ge & \eps\cdot\left(f\left(AB\right)+f\left(BC\right)\right)\ge\eps\cdot f\left(ABC\right).
\end{alignedat}
\]
\end{proof}
Recall that $m\left(S\right)$ and $VC\left(S\right)$ are the maximum-matching
size and the vertex-cover size, respectively, in the graph defined
by the stream $S$. Although they are both not smooth functions (as
shown in Corollary \ref{cor:MM=000026VC2AS}), they are almost smooth
(as proved by Crouch et al. \cite{Crouch2013} for $m\left(\cdot\right)$,
and reproduced here for completeness).
\begin{cor}
\label{cor:MM=000026VC2AS}The maximum-matching size $m\left(\cdot\right)$
and the minimum vertex-cover size $VC\left(\cdot\right)$ are both
$2$-almost-smooth. Moreover, they are both not $d$-almost-smooth
for any $d<2$.
\end{cor}

\begin{proof}
Obviously both $m\left(\cdot\right)$ and $VC\left(\cdot\right)$
are non-negative, bounded and monotone, since on a longer segment
of the stream both the maximum-matching and the minimum vertex-cover
cannot be smaller. Hence, we are only left to show that they are both
subadditive.

Let $M$ be a maximum-matching of the graph defined by the stream
$AB$, and denote by $M_{A}$ and $M_{B}$ the edges from $M$ that
appear in $A$ and $B$, respectively. Note that $M_{A}$ is a matching
in the graph defined by the stream $A$, and similarly for $M_{B}$.
Thus, clearly $\left|M_{A}\right|\le m\left(A\right)$ and $\left|M_{B}\right|\le m\left(B\right)$,
and therefore
\[
m\left(AB\right)=\left|M\right|=\left|M_{A}\right|+\left|M_{B}\right|\le m\left(A\right)+m\left(B\right),
\]
and so $m\left(\cdot\right)$ is subadditive.

Observe that for a disjoint segments of the stream $A$ and $B$,
the union of a minimum vertex-cover on $A$ and a minimum vertex-cover
on $B$ is clearly a feasible (not necessarily minimum) vertex-cover
on $AB$, and since it is a minimization problem we obtain $VC\left(A\right)+VC\left(B\right)\ge VC\left(AB\right)$.
Hence $VC\left(\cdot\right)$ is also subadditive.

For the tightness argument we use the following tight example. Let
$\G$ be a graph composed of $n$ vertex-disjoint paths of length
$3$, i.e., $n$ paths of the form $e_{a}=\left\{ x,y\right\} ,e_{b}=\left\{ y,z\right\} ,e_{c}=\left\{ z,w\right\} $.
The segment $A$ of the stream contains all the $e_{a}$ edges, $B$
contains all the $e_{b}$ edges, and $C$ contains all the $e_{c}$
edges. Obviously $m\left(AB\right)=m\left(B\right)=m\left(BC\right)=n$
while $m\left(ABC\right)=2n$, and similarly for $VC\left(\cdot\right)$.
In particular, both maximum-matching size and minimum vertex-cover
are not smooth as per the original definition of \cite{Braverman:2007:SHS:1333875.1334202}.
\end{proof}
~
\begin{rem}
It is easy to see that Corollary \ref{cor:MM=000026VC2AS} holds even
for hypergraphs by the same arguments. 
\end{rem}

We analyze the smooth-histogram algorithm of \cite{Braverman:2007:SHS:1333875.1334202}
for functions that are almost-smooth with constant approximation ratio.

\begin{thm}
\label{thm:Almost-Smooth} [Formal version of Theorem \ref{thm:Almost-Smooth-1}]
Let $f$ be a $\left(c,d\right)$-almost-smooth function defined on
streams. Assume that for every $\eps,\delt\in\left(0,\frac{1}{2}\right)$,
there exists an algorithm $\Lambda$ for insertion-only streams that
$\left(\left(1+\eps\right)C,\delt\right)$-approximates $f$ using
space $s\left(\eps,\delt\right)$ and update time $t\left(\eps,\delt\right)$.
Then for every $\eps,\delt\in\left(0,\frac{1}{2}\right)$ there exists
a sliding-window algorithm $\lsw$ that\\
 $\left(dc^{2}C^{2}\left(1+\O{\eps}\right),\delt\right)$-approximates
$f$ using space $\O{\eps^{-1}\log w\cdot\left(s\left(\eps,\frac{\eps\delt}{2w\log w}\right)+\log w\right)}$
and update time $\O{\eps^{-1}\log w\cdot t\left(\eps,\frac{\eps\delt}{2w\log w}\right)}$.
\end{thm}

We prove Theorem \ref{thm:Almost-Smooth} in appendix \ref{appendix:Proof-of-Main-Theorem}.
At a high level, we adapt the approach and notations of Crouch et
al. \cite{Crouch2013}, which in turns is based on the smooth-histogram
method of Braverman and Ostrovsky \cite{Braverman:2007:SHS:1333875.1334202}.

For certain approximation algorithms we can reduce the dependence
on the approximation factor $C$ from quadratic $\left(C^{2}\right)$
to linear $\left(C\right)$. Suppose that the approximation algorithm
$\Lambda$ of the function $f$ has the following form: It $\left(1+\eps,\delt\right)$-approximates
a function $g$, and this $g$ is a $C$-approximation of $f$. Now,
if $g$ itself is $\left(c,d\right)$-almost-smooth then we can save
a factor of $C$ by arguing directly about approximating $g$.
\begin{thm}
\label{thm:Ext-Almost-Smooth} Let $f$ be some function, let $g$
be a $\left(c,d\right)$-almost-smooth function, and assume that $g$
is a $C$-approximation of $f$. Assume that for every $\eps,\delt\in\left(0,\frac{1}{2}\right)$,
there exists an algorithm $\Lambda$ for insertion-only streams that
$\left(1+\eps,\delt\right)$-approximates $g$ using space $s\left(\eps,\delt\right)$
and update time $t\left(\eps,\delt\right)$. Then for every $\eps,\delt\in\left(0,\frac{1}{2}\right)$
there exists a sliding-window algorithm $\lsw$ that $\left(dc^{2}C\left(1+\O{\eps}\right),\delt\right)$-approximates
$f$ using space $\O{\eps^{-1}\log w\cdot\left(s\left(\eps,\frac{\eps\delt}{2w\log w}\right)+\log w\right)}$
and update time $\O{\eps^{-1}\log w\cdot t\left(\eps,\frac{\eps\delt}{2w\log w}\right)}$.
\end{thm}

\begin{proof}
By applying Theorem \ref{thm:Almost-Smooth} to the function $g$
and the algorithm $\Lambda$, that approximates it, we obtain a sliding-window
algorithm $\lsw$ that computes a $\left(dc^{2}\left(1+\O{\eps}\right),\delt\right)$-approximation
of $g$, uses space $\O{\eps^{-1}\log w\cdot\left(s\left(\eps,\frac{\eps\delt}{2w\log w}\right)+\log w\right)}$
and update time $\O{\eps^{-1}\log w\cdot t\left(\eps,\frac{\eps\delt}{2w\log w}\right)}$.
Since $g$ is a $C$-approximation of $f$, this algorithm $\lsw$
is in fact a $\left(dc^{2}C\left(1+\O{\eps}\right),\delt\right)$-approximation
of $f$, using the same space and update time.
\end{proof}

\section{\label{sec:AS-example}Example Applications of Almost-Smoothness}

We present a variety of examples where our framework yields new approximation
algorithms with relatively little effort. More specifically, we provide
examples of functions $f$ (over a stream) that are $\O 1$-almost-smooth,
and then employ our framework to convert known algorithms for insertion-only
streams into new algorithms for sliding-window streams, with a slightly
larger approximation ratio. For all these problems, the smooth-histogram
technique of Braverman and Ostrovsky \cite{Braverman:2007:SHS:1333875.1334202}
is not applicable, because these functions are not known to be smooth.
In fact, many of them are provably not smooth, which justifies the
necessity of our generalized framework.

Our first example (Section \ref{subsec:Symmetric-Norms}) is for the
common model of a frequency vector (additive updates to a vector),
where the function $f$ to be computed is an arbitrary symmetric norm
of that vector. We show that such a norm is subadditive and thus $2$-almost-smooth,
and then use a known algorithm \cite{BBC17} to obtain a sliding-window
algorithm. Our second example (Section \ref{subsec:Matrix-Streams})
is for streaming matrices where each item in the stream is the (next)
row of the matrix; this is a natural model for reading a matrix in
row order. We show that the Schatten $p$-norm of the matrix is $\sqrt{2}$-almost-smooth,
for all $p\ge2$, and we use a known algorithm for $p=4$ \cite{Braverman2020SchattenNI}
to derive a new sliding-window algorithm.

We then consider some discrete problems (Sections \ref{subsec:Submodular-Functions}
and \ref{subsec:Max-k-Cover}). One of them is maximum submodular
matching in a graph defined by a stream of edges on fixed vertex set
$\left[n\right]$. Using that every submodular function is subadditive,
we obtain $2$-almost-smoothness and use an algorithm of \cite{Chaka14}
to conclude a sliding-window algorithm. Another discrete problem is
max $k$-cover in a graph, for which we show both approximation and
estimation sliding-window algorithms. We use our result for symmetric
norms to show a $4$-estimation sliding-window algorithm for max $k$-cover
using $\Oet{\frac{n}{k}}$ bits of space. Lastly, we design for every
$d\in\left(1,2\right]$ an artificial function, based on the maximum-matching
size, that is $d$-almost smooth, and show it is tight.

\subsection{\label{subsec:Symmetric-Norms}Symmetric Norms}

Consider the frequency-vector streaming model, where the stream is
composed of additive updates to an underlying $n$-dimensional vector
$x\in\R^{n}$. We assume that all updates are positive and all entries
are polynomially bounded, i.e., $\left|x_{i}\right|\le\poly\left(n\right)$
for all $i\in\left[n\right]$. The assumption of positive updates
is necessary because otherwise even $2$-approximation of the number
of distinct items in a stream requires $\Omega\left(w\right)$ bits
of space, see Datar et al. \cite{StreamStatistics}.

We show that every symmetric norm $\ell:\R^{n}\rightarrow\R$ (symmetric
means invariant under sign-flips and coordinate-permutations, see
\cite{bhatia97}) can be $\left(2+\eps\right)$-approximated in the
sliding-window model with $\O{\mmc\left(\ell\right)^{2}\poly\left(\frac{\log nw}{\eps}\right)}$
bits of space, where $\mmc\left(\cdot\right)$ is the norm's maximum
modulus of concentration (see \cite{BBC17}). It is known that $\mmc\left(\ell\right)$
depends only on the norm $\ell$ itself and that $1\le\mmc\left(\ell\right)\le\sqrt{n}$.
 For a concrete example of a norm that is $2$-almost-smooth and this
parameter $d=2$ is tight, we consider the top-$k$ norm (see Definition
\ref{top-k_norm}) for $k=\Theta\left(n\right)$, and deduce for it
a sliding-window $\left(2+\eps\right)$-approximation algorithm with
$\poly\left(\frac{\log n}{\eps}\right)$ bits of space.
\begin{lem}
\label{lem:norm_subadditive}Every monotone norm is subadditive, and
hence $2$-almost-smooth.
\end{lem}

\begin{proof}
Subadditivity follows immediately from the triangle inequality, and
obviously every norm is non-negative and bounded. We assume monotonicity,
and therefore Lemma \ref{lem:subadditive2AS} implies that the norm
is $2$-almost-smooth.
\end{proof}
B$\l$asiok et al. \cite{BBC17} showed a $\left(1+\eps\right)$-approximation
randomized streaming algorithm for an arbitrary symmetric norm. We
use their algorithm to achieve $\left(2+\eps\right)$-approximation
sliding-window algorithm with almost the same space complexity.
\begin{thm}
[Formal version of Theorem \ref{thm:symmetric_norm-1}] \label{thm:symmetric_norm}Let
$\ell$ be a symmetric norm on $\R^{n}$. There exists a $\left(2+\eps\right)$-approximation
sliding-window randomized algorithm computing the norm $\ell$ using
$\O{\mmc\left(\ell\right)^{2}\poly\left(\frac{\log nw}{\eps}\right)}$
bits of space with constant success probability.
\end{thm}

\begin{proof}
Algorithm $1$ of \cite{BBC17} is a $\left(1+\eps\right)$-approximation
of a symmetric norm $\ell$ of a vector defined by a stream using
$\O{\mmc\left(\ell\right)^{2}\poly\left(\frac{\log n}{\eps}\right)}$
bits of space. This algorithm has constant success probability, which
can be amplified to $1-\frac{\eps}{10w\log w}$ using the median of
$\O{\log\frac{w}{\eps}}$ independent repetitions. Every symmetric
norm is monotone (Proposition IV.1.1 in \cite{bhatia97}), therefore
$\ell$ is $2$-almost-smooth (Lemma \ref{lem:norm_subadditive}).
This implies, by Theorem \ref{thm:Almost-Smooth}, a $\left(2+\eps\right)$-approximation
sliding-window algorithm using $\O{\mmc\left(\ell\right)^{2}\poly\left(\frac{\log nw}{\eps}\right)}$
bits of space.
\end{proof}
We now show a norm for which our almost-smooth framework is applicable,
but the smooth-histogram technique of Braverman and Ostrovsky \cite{Braverman:2007:SHS:1333875.1334202}
does not.
\begin{defn}
\label{top-k_norm}For $1\le k\le n$ the top-$k$ norm of a vector
$x\in\R^{n}$, denoted by $\left\Vert x\right\Vert _{\left(k\right)}$,
is the $\ell_{1}$ norm of the $k$ entries of $x$ with the largest
absolute value, i.e. $\left\Vert x\right\Vert _{\left(k\right)}=\sum\limits _{i=1}^{k}\left|x_{i}^{\downarrow}\right|$,
where $x^{\downarrow}$ is the vector $x$ reordered in decreasing
magnitude.
\end{defn}

\begin{rem}
\label{rem:top-k_norm_tight}Since the top-$k$ norm is monotone,
Lemma \ref{lem:norm_subadditive} shows that it is $2$-almost-smooth.
Furthermore, for $k\le\frac{n}{2}$, the almost-smoothness parameter
$d=2$ of the top-$k$ norm is tight. Let $A$ be the stream $\left(e_{1},\cdots,e_{k}\right)$,
where each $e_{i}\in\R^{n}$ is the $i$-th standard basis vector,
thus its total additive update is $x^{A}=\sum\limits _{i=1}^{k}e_{i}$.
Similarly, $B$ is the stream $\left(e_{k+1},\cdots,e_{n}\right)$
with corresponding vector $x^{B}=\sum\limits _{i=k+1}^{n}e_{i}$,
and $C$ has the same updates as stream $A$, thus $x^{C}=x^{A}$.
Observe that $\left\Vert x^{A}+x^{B}\right\Vert _{\left(k\right)}=\left\Vert x^{B}\right\Vert _{\left(k\right)}=k$
but $\left\Vert x^{A}+x^{B}+x^{C}\right\Vert _{\left(k\right)}=2k$
while $\left\Vert x^{B}+x^{C}\right\Vert _{\left(k\right)}=k$, hence
$d=2$ is indeed tight in the almost-smoothness of the top-$k$ norm.
\end{rem}

The maximum modulus of concentration of the top-$k$ norm is $\mmc\left(\left\Vert \cdot\right\Vert _{\left(k\right)}\right)=\Ot{\sqrt{\frac{n}{k}}}$,
see Section $6.1$ of \cite{BBC17}. Therefore, as a corollary of
Theorem \ref{thm:symmetric_norm} for the top-$k$ norm, we obtain
the following sliding-window algorithm.

\begin{cor}
\label{cor:top-k_SW}There exists a $\left(2+\eps\right)$-approximation
sliding-window randomized algorithm for the top-$k$ norm with $\O{\frac{n}{k}\poly\left(\frac{\log n}{\eps}\right)}$
bits of space. Specifically, for $k=\Theta\left(n\right)$, the space
complexity is $\poly\left(\frac{\log n}{\eps}\right)$.
\end{cor}

\begin{rem}
We are not aware of a lower bound on the approximation ratio for sliding-window
algorithms, and the approximation ratio in Corollary \ref{cor:top-k_SW}
can possibly be improved, say to $1+\eps$, using the same space complexity.
\end{rem}

\subsection{\label{subsec:Matrix-Streams}Matrix Streams}

A $\emph{matrix stream}$ is a sequence of updates to an underlying
matrix $X\in\R^{n\times m}$, for $n\ge m$, initialized to the all-zeros
matrix. We consider only row-order streams, where each update adds
a new row to the matrix (from row $1$ to row $n$), hence the $i$-th
update in the stream is a row vector $X_{i}\in\R^{m}$ representing
the $i$-th row of $X$. We assume that all entries are polynomially
bounded, i.e., $\left|X_{i,j}\right|\le\poly\left(n\right)$ for all
$i,j$.

We are interested in the $\emph{Schatten \ensuremath{p}-norm}$ of
a matrix $X$, i.e., the $\ell_{p}$-norm of the spectrum of $X$.
\begin{defn}
The $\emph{Schatten \ensuremath{p}-norm}$, for $p\ge1$, of a matrix
$X\in\R^{n\times m}$ with singular values $\sigma_{1}\ge\cdots\ge\sigma_{m}\ge0$,
is defined as
\[
\left\Vert X\right\Vert _{S_{p}}=\left(\sum_{i=1}^{m}\sigma_{i}^{p}\right)^{\nicefrac{1}{p}}.
\]
\end{defn}

\begin{rem}
Note that $\left\Vert X\right\Vert _{S_{2p}}^{2}=\left\Vert X^{T}X\right\Vert _{S_{p}}$
since $X^{T}X$ is a PSD matrix with eigenvalues $\left\{ \sigma_{i}^{2}\mid i\in\left[m\right]\right\} $.
\end{rem}

\begin{cor}
For $p\ge2$, the Schatten $p$-norm is $\sqrt{2}$-almost-smooth.
\end{cor}

\begin{proof}
Let $X$ be a matrix defined by a stream of $n$ rows $X_{1},\dots,X_{n}$.
The first $n'$ rows define a matrix $Y$ and the other $n-n^{'}$
rows define $Z$. Observe that $X^{T}X=\sum\limits _{i=1}^{n}X_{i}^{T}X_{i}=\sum\limits _{i=1}^{n^{'}}X_{i}^{T}X_{i}+\sum\limits _{i=1}^{n-n^{'}}X_{i}^{T}X_{i}=Y^{T}Y+Z^{T}Z$,
and using the triangle inequality for the Schatten $p/2$-norm (recall
$p\ge2$), we obtain
\[
\left\Vert X\right\Vert _{S_{p}}^{2}=\left\Vert Y^{T}Y+Z^{T}Z\right\Vert _{S_{p/2}}\le\left\Vert Y^{T}Y\right\Vert _{S_{p/2}}+\left\Vert Z^{T}Z\right\Vert _{S_{p/2}}=\left\Vert Y\right\Vert _{S_{p}}^{2}+\left\Vert Z\right\Vert _{S_{p}}^{2},
\]
hence the squared Schatten $p$-norm is subadditive. It is clearly
also non-negative, thus, in order to prove almost-smoothness we only
need to show it is bounded and monotone. Every rank-$1$ matrix $v^{T}v$,
for a row vector $v\in\R^{m}$, has (at most) one eigenvalue $\lambda=\left\Vert v\right\Vert _{2}^{2}$
different from $0$, hence $\left\Vert v^{T}v\right\Vert _{S_{p/2}}=\left\Vert v\right\Vert _{2}^{2}$,
and therefore
\[
\left\Vert X\right\Vert _{S_{p}}^{2}=\left\Vert X^{T}X\right\Vert _{S_{\nicefrac{p}{2}}}\le\sum\limits _{i=1}^{n}\left\Vert X_{i}^{T}X_{i}\right\Vert _{S_{\nicefrac{p}{2}}}=\sum\limits _{i=1}^{n}\left\Vert X_{i}\right\Vert _{2}^{2}\le\poly\left(n\right).
\]
Thus, the squared Schatten $p$-norm is bounded. To prove it is monotone,
we use Weyl's inequality, that for every real matrices $Y,Z$ it holds
that $\sigma_{i}\left(Y^{T}Y+Z^{T}Z\right)\ge\sigma_{i}\left(Y^{T}Y\right)$
for all $i$ and thus $\left\Vert Y+Z\right\Vert _{S_{p}}^{2}\ge\left\Vert Y\right\Vert _{S_{p}}^{2}$.

It then follows from Lemma \ref{lem:subadditive2AS} that the squared
Schatten $p$-norm is $2$-almost-smooth. It is then straightforward
that the Schatten $p$-norm itself is $\sqrt{2}$-almost-smooth. 
\end{proof}
By plugging a known streaming algorithm for the Schatten $4$-norm
into our almost-smooth framework, we obtain a sliding-window algorithm
with constant approximation ratio and constant space.
\begin{cor}
[Formal version of Theorem \ref{thm:schatten_4-norm_alg-1}] \label{cor:schatten_4-norm_alg}
There exists a sliding-window algorithm that $\left(\sqrt{2}+\eps\right)$-approximates
the Schatten $4$-norm using $\O{\eps^{-3}\log n\log\frac{w}{\eps}\log w}$
bits of space with constant success probability.
\end{cor}

\begin{proof}
Algorithm $5$ of \cite{Braverman2020SchattenNI} is a $\left(1+\eps\right)$-approximation
of the Schatten $4$-norm in the row-order model, using $\O{\eps^{-2}\log n}$
bits of space. This algorithm has constant success probability, which
can be amplified to $1-\frac{\eps}{10w\log w}$ using the median of
$\O{\log\frac{w}{\eps}}$ independent repetitions. Now Theorem \ref{thm:Almost-Smooth}
and the above $\sqrt{2}$-almost-smoothness of the Schatten $4$-norm
imply a $\left(\sqrt{2}+\eps\right)$-approximation sliding-window
algorithm using $\O{\eps^{-3}\log n\log\frac{w}{\eps}\log w}$ bits
of space.
\end{proof}

\subsection{\label{subsec:Submodular-Functions}Submodular Functions}

The $\emph{maximum submodular matching }$ (MSM) problem is a generalization
of the maximum weight matching (MWM) problem to submodular objective
functions.
\begin{defn}
A $\emph{submodular function}$ on a ground set $\mathcal{X}$ is
a set function $f:2^{\mathcal{X}}\rightarrow\R$ that satisfies $f\left(A\cup B\right)+f\left(A\cap B\right)\le f\left(A\right)+f\left(B\right)$
for all $A,B\sseq\mathcal{X}$.
\end{defn}

\begin{defn}
For a graph $\G$ and a monotone non-negative submodular function
$f:2^{E}\rightarrow\R_{+}$, the MSM problem asks to find a matching
$M\sseq E$ such that $f\left(M\right)$ is maximized. If $M^{*}\sseq E$
is a matching maximizing $f\left(\cdot\right)$ we denote $\widehat{f}\left(G\right)=f\left(M^{*}\right)$,
and say that $M^{*}$ is an $f$-maximum matching. For a stream $S$
of edges whose underlying graph is $G$, we define $\widehat{f}\left(S\right)=\widehat{f}\left(G\right)$.
\end{defn}

\begin{rem}
\label{rem: submodularTOsubadditive}If $f:2^{E}\rightarrow\R_{+}$
is submodular then the derived function $\widehat{f}$ of streams
of edges is subadditive (see Definition \ref{def:subadditive}). To
see this, let $AB$ be a stream with $f$-maximum matching $M^{*}$.
Denote by $M_{A}^{*}$ and $M_{B}^{*}$ the edges from the matching
$M^{*}$ that appear in the sub-streams $A$ and $B$, respectively.
Since $A$ clearly contains $M_{A}^{*}$, the $f$-maximum matching
in $A$ has value $\widehat{f}\left(A\right)\ge f\left(M_{A}^{*}\right)$,
and similarly for $B$. Therefore,
\[
\widehat{f}\left(AB\right)=f\left(M^{*}\right)\le f\left(M_{A}^{*}\right)+f\left(M_{B}^{*}\right)-f\left(\emptyset\right)\le\widehat{f}\left(A\right)+\widehat{f}\left(B\right).
\]
\end{rem}

Chakrabarti and Kale \cite{Chaka14} showed a constant-factor approximation
algorithm for the MSM problem in the semi-streaming model. They modeled
access to $f$ using a value oracle, which upon receiving a query
subset $E^{'}\sseq E$ outputs $f\left(E^{'}\right)$. We use the
same model of a value oracle. By apply our almost-smooth framework
to their algorithm, we achieve a constant-factor approximation algorithm
for the MSM problem in the sliding-window model.
\begin{cor}
\label{cor:MSM_alg}There exists a deterministic $\O 1$-approximation
sliding-window algorithm for the MSM problem with $\O{n\log n\log w}$
bits of space.
\end{cor}

\begin{proof}
Theorem $1$ of \cite{Chaka14} provides a deterministic $7.75$-approximation
insertion-only algorithm $\Lambda$ for the MSM problem with $\O{n\log n}$
bits of space. Since the MSM problem is subadditive (see Remark \ref{rem: submodularTOsubadditive}),
and thus $2$-almost-smooth, we derive from Theorem \ref{thm:Almost-Smooth}
a sliding-window algorithm $\Lambda^{sw}$ that $\left(2\cdot7.75^{2}\left(1+\eps\right)\right)$-approximates
the MSM problem using $\O{\eps^{-1}n\log n\log w}$ bits of space.
\end{proof}
\begin{rem}
One interesting submodular function follows from the Facility Location
problem. Let $m\in\N$ and $\G$ be a graph, and for every edge $e\in E$
let $V_{e}\in\R_{+}^{m}$ be a vector of values associated with the
edge $e$. The function $f:2^{E}\rightarrow\R_{+}$ defined by $f\left(E^{'}\right)=\sum\limits _{i=1}^{m}\max\limits _{e\in E^{'}}V_{e,i}$
for $E^{'}\sseq E$ is submodular, monotone, and non-negative (See
\cite{10.1007/BF01585521,KG11}).
\end{rem}

\subsection{\label{subsec:Max-k-Cover}Max $k$-Cover}

The $\emph{Max \ensuremath{k}-Cover}$ problem is dual to Vertex-Cover,
and its goal is to cover as many edges as possible (in a graph $G$)
using only $k$ vertices. Recently, McGregor and Vu \cite{McGregorAndTVu}
studied the Max $k$-Cover problem in the dynamic graph stream model,
where the stream consists of edge insertions and deletions. Specifically,
they showed a $\left(1+\eps\right)$-approximation streaming algorithm
for arbitrary $k$ with space bound $\Ot{\eps^{-2}n}$ and success
probability $1-\frac{1}{\poly\left(n\right)}$, where $n=\left|V\right|$,
see Theorem $20$ of \cite{McGregorAndTVu}. Note that we can assume
the entire stream is of length at most $2w$ (see claim \ref{claim:SW_bound_by_2w})
and since we are only concerned with insertion-only graph streams
we assume $w=\O{n^{2}}$. The aforementioned algorithm already has
high probability of success and thus requires no amplification. Using
this and our almost-smoothness technique, we achieve $\left(2+\eps\right)$-approximation
algorithm in the sliding-window model, with space bound $\Ot{\eps^{-3}n}$.
In addition, we present a $\left(2+\eps\right)$-estimation algorithm
in the same dynamic graph stream model, and consequently a $\left(4+\eps\right)$-estimation
algorithm in the sliding-window model, both with space bound $\Oet{\frac{n}{k}}$.
\begin{defn}
For a graph $\G$, a subset of vertices $U\sseq V$ is said to $\emph{cover}$
the set of incident edges $\widetilde{E}\left(U\right)=\left\{ e\in E:e\cap U\neq\emptyset\right\} $.
Denote by $c_{k}\left(G\right)=\max\left\{ \left|\widetilde{E}\left(U\right)\right|:\left|U\right|=k\right\} $
the largest number of edges that can be covered by $k$ vertices in
$G$. For a stream $S$ of edges, whose underlying graph is $G$,
we define $c_{k}\left(S\right)=c_{k}\left(G\right)$.
\end{defn}

\begin{cor}
For every $k$, the Max $k$-Cover function $c_{k}\left(\cdot\right)$
is subadditive, and hence $2$-almost-smooth.
\end{cor}

\begin{proof}
Let $AB$ be a stream of edges that defines a graph $\G$ and suppose
$c_{k}\left(G\right)=\left|\widetilde{E}\left(U\right)\right|$ for
some $U\sseq V$ with $\left|U\right|=k$. Let $E_{A}=\widetilde{E}\left(U\right)\cap A$
and $E_{B}=\widetilde{E}\left(U\right)\cap B$ and observe that $\left|\widetilde{E}\left(U\right)\right|=\left|E_{A}\right|+\left|E_{B}\right|$.
Therefore
\[
c_{k}\left(AB\right)=\left|\widetilde{E}\left(U\right)\right|=\left|E_{A}\right|+\left|E_{B}\right|\le c_{k}\left(A\right)+c_{k}\left(B\right),
\]
and so $c_{k}\left(\cdot\right)$ is subadditive. Obviously $c_{k}\left(\cdot\right)$
is also non-negative, bounded and monotone, hence Lemma \ref{lem:subadditive2AS}
applies.
\end{proof}
We start with a sliding-window approximation algorithm for Max $k$-Cover.
We obtain a sliding-window algorithm by plugging a known algorithm
into Theorem \ref{thm:Almost-Smooth}. Specifically, we use the aforementioned
algorithm of McGregor and Vu \cite[Theorem 20]{McGregorAndTVu} and
the fact that Max $k$-Cover is $2$-almost-smooth to obtain the following
corollary.
\begin{cor}
\label{cor:max-k-cover_apprx_alg}For every $k$ there exists a sliding-window
$\left(2+\eps\right)$-approximation algorithm for Max $k$-Cover
using $\Ot{\eps^{-3}n}$ bits of space.
\end{cor}

We now design a sliding-window estimation streaming algorithms for
Max $k$-Cover. Observe that if $\sigma_{k}\left(G\right)$ denotes
the sum of degrees of the $k$ vertices of largest degree in $G$,
then a greedy algorithm that reports $\sigma_{k}\left(G\right)$,
achieves a $2$-estimation of Max $k$-Cover. We can implement it
in a graph stream using a known algorithm for the top-$k$ norm (similarly
to Corollary \ref{cor:top-k_SW}),  which will $\left(2+\eps\right)$-estimate
$\sigma_{k}\left(G\right)$. Hence, we deduce a $\left(2+\eps\right)$-estimation
algorithm for Max $k$-Cover in the dynamic graph stream model with
space bound $\Oet{\frac{n}{k}}$. Moreover, we can similarly obtain
a $\left(4+\eps\right)$-estimation algorithm in the sliding-window
model using a sliding-window algorithm for the top-$k$ norm. The
details follow.
\begin{thm}
\label{thm:max-k-cover_alg}For every $k$ there exists a sliding-window
$\left(4+\eps\right)$-estimation algorithm for Max $k$-Cover using
$\Oet{\frac{n}{k}}$ bits of space.
\end{thm}

\begin{proof}
Let $x\in\R^{n}$ represent the degrees of the $n$ vertices, and
notice that $\left\Vert x\right\Vert _{\left(k\right)}=\sigma_{k}\left(G\right)$
using the notation from Section \ref{subsec:Symmetric-Norms}. By
incrementing coordinates $x_{i},x_{j}$ whenever an edge $\left(i,j\right)$
arrives, we create a virtual stream (twice as long) simulating updates
to $x$ in the frequency-vector model. Using Corollary \ref{cor:top-k_SW}
(with twice as large window size) we can $\left(2+\eps\right)$-estimate
$\sigma_{k}\left(G\right)$ in the sliding-window model using $\Oet{\frac{n}{k}}$
bits of space. Since $\frac{1}{2}c_{k}\left(G\right)\le\frac{1}{2}\sigma_{k}\left(G\right)\le c_{k}\left(G\right)$,
reporting the $\left(2+\eps\right)$-estimate of $\sigma_{k}\left(G\right)$
gives a sliding-window $\left(4+\eps\right)$-estimate of Max $k$-Cover
and uses $\Oet{\frac{n}{k}}$ bits of space.
\end{proof}

\subsection{Example of $d$-Almost-Smooth Functions for $d\in\left(1,2\right]$}

We design a special family of functions whose almost-smoothness parameter
$d$ can be any number between $1$ and $2$. For every $d\in\left(1,2\right]$
define the following function for an $n$-vertex graph $G$ with maximum
matching size $m\left(G\right)$
\[
f_{d}\left(G\right)=\frac{1}{4}\left(\frac{2-d}{d-1}\right)n+m\left(G\right).
\]
For a stream $S$ of edges, with an underlying graph $G$, we define
$f_{d}\left(S\right)=f_{d}\left(G\right)$.
\begin{lem}
\label{lem:f_d-d-AS}For every $d\in\left(1,2\right]$ the function
$f_{d}$ defined above is $d$-almost-smooth and not $d^{\prime}$-almost-smooth
for any $d^{\prime}<d$.
\end{lem}

\begin{proof}
Let $d\in\left(1,2\right]$, and observe that the function $f_{d}\left(\cdot\right)$
is non-negative, bounded and monotone, since $m\left(\cdot\right)$
satisfies all those requirements. Hence, we are only left to show
the almost-smoothness property. We shall use the equivalent definition
explained in Remark \ref{rem:ASfunctionRem}. Define $\alpha_{d}=\frac{1}{4}\left(\frac{2-d}{d-1}\right)$,
let $\eps\in\left(0,1\right]$ and let $A,B$ and $C$ be disjoint
segments of the stream satisfying $\eps f_{d}\left(AB\right)\le f_{d}\left(B\right)$,
i.e., $\eps\alpha_{d}n+\eps m\left(AB\right)\le\alpha_{d}n+m\left(B\right)$.
We divide the proof into two cases. If $m\left(BC\right)\le\frac{n}{4}=\frac{d-1}{2-d}\alpha_{d}n$
then using the same logic as in Lemma \ref{lem:subadditive2AS} we
obtain
\[
\begin{alignedat}{2}2m\left(BC\right) & \ge m\left(BC\right)+m\left(B\right)\\
 & \ge m\left(BC\right)+\eps\alpha_{d}n+\eps m\left(AB\right)-\alpha_{d}n &  & \text{since }\eps f_{d}\left(AB\right)\le f_{d}\left(B\right)\\
 & \ge\eps\left(m\left(AB\right)+m\left(BC\right)\right)+\eps\alpha_{d}n-\alpha_{d}n & \hspace{1cm} & \text{since }\eps\le1\\
 & \ge\eps m\left(ABC\right)+\eps\alpha_{d}n-\alpha_{d}n. &  & \text{since }m\left(\cdot\right)\text{ is subadditive}
\end{alignedat}
\]
Hence, $2m\left(BC\right)+\alpha_{d}n\ge\eps m\left(ABC\right)+\eps\alpha_{d}n=\eps f_{d}\left(ABC\right)$,
and therefore
\[
\eps f_{d}\left(ABC\right)\le\left(d+2-d\right)m\left(BC\right)+\alpha_{d}n\le dm\left(BC\right)+\left(d-1\right)\alpha_{d}n+\alpha_{d}n=d\cdot f_{d}\left(BC\right).
\]
In the other case, when $m\left(BC\right)>\frac{n}{4}$, since the
graph has $n$ vertices we have $m\left(ABC\right)\le\frac{n}{2}$,
and therefore
\[
\eps f_{d}\left(ABC\right)\le f_{d}\left(ABC\right)=\alpha_{d}n+m\left(ABC\right)\le\alpha_{d}n+\frac{n}{2}.
\]
From the assumption $m\left(BC\right)>\frac{n}{4}$ we obtain
\[
d\cdot f_{d}\left(BC\right)=d\cdot\alpha_{d}n+d\cdot m\left(BC\right)\ge d\cdot\alpha_{d}n+d\cdot\frac{n}{4}=\alpha_{d}n+\frac{n}{2}+\left(d-1\right)\cdot\alpha_{d}n+\frac{dn}{4}-\frac{n}{2}=\alpha_{d}n+\frac{n}{2},
\]
hence $\eps f_{d}\left(ABC\right)\le d\cdot f_{d}\left(BC\right)$,
as required.

For the tightness argument we use the same example as in Corollary
\ref{cor:MM=000026VC2AS} of disjoint union of $n$ paths of length
$3$. Recall that $\left|V\right|=4n$ and thus $m\left(AB\right)=m\left(B\right)=m\left(BC\right)=n$
but $m\left(ABC\right)=2n$, therefore we deduce that
\[
\frac{f_{d}\left(ABC\right)}{f_{d}\left(BC\right)}=\frac{\frac{1}{4}\left(\frac{2-d}{d-1}\right)4n+2n}{\frac{1}{4}\left(\frac{2-d}{d-1}\right)4n+n}=\frac{\frac{2-d}{d-1}+2}{\frac{2-d}{d-1}+1}=\frac{2-d+2d-2}{2-d+d-1}=d,
\]
showing that indeed $f_{d}$ is not $d^{\prime}$-almost-smooth for
any $d^{\prime}<d$.
\end{proof}

\section{\label{sec:Maximum-Matching}Applications to Maximum-Matching}

We show here a concrete example of the usefulness of the almost-smooth
histogram framework for the graph streaming model. Specifically, for
graphs of bounded arboricity $\alpha$, we use a known insertion-only
$\O{\alpha}$-estimation algorithm for maximum-matching, and deduce
a sliding-window algorithm with approximation factor $\O{\alpha^{2}}$
and space $\polylog\left(n\right)$. We then improve the approximation
ratio to $\O{\alpha}$ by observing that the number of $\alpha$-good
edges (the quantity used to approximate the maximum-matching size)
is itself a subadditive function, and thus we can argue directly about
it. See Table \ref{tab:SWMM}.

Recall that in the usual graph streaming model, the input is a stream
of edge insertions to an underlying graph on the set of vertices $V=\left[n\right]$,
where $n$ is known in advance. We assume that the underlying graph
does not contain parallel edges, i.e., the stream of edges does not
contain the same edge twice. Hence, the length of the entire stream
is bounded by $n^{2}$.

In the sliding-window model the graph is defined using only the last
$w$ edge insertions from the stream, referred to as the active window
$W$. Note that $w$ is known (to the algorithm) in advance, and that
$w\le n^{2}$, as the length of the entire stream is bounded by $n^{2}$.

McGregor and Vorotnikova \cite{mcgregor_et_al:OASIcs:2018:8295},
based on the result of Cormode et al. \cite{cormode_et_al:LIPIcs:2017:7849},
presented an algorithm that approximates the size of the maximum-matching
in a graph with arboricity $\alpha$ within factor $\left(1+\eps\right)\left(\alpha+2\right)$,
with constant probability, using space $\O{\eps^{-2}\log^{2}n}$ and
update time $\O{\eps^{-2}\log n}$. To achieve low failure probability
$\delt$ it is standard to compute a median of $\log\delt^{-1}$ parallel
repetitions. Therefore, direct application of the Almost-Smooth-Histogram
method yields the following theorem (but see Theorem \ref{thm:improved_thm}
for the improved approximation).
\begin{thm}
\label{thm:a^2-approx}For every $\eps,\delt\in\left(0,\frac{1}{2}\right)$,
there is a sliding-window $\left(\left(2+\eps\right)\left(\alpha+2\right)^{2},\delt\right)$-estimation
algorithm for maximum-matching size in a graph with arboricity $\alpha$,
with space complexity $\O{\eps^{-3}\log^{4}n\log\tfrac{1}{\eps\delt}}$
and update time $\O{\eps^{-3}\log^{3}n\log\tfrac{1}{\eps\delt}}$.
\end{thm}

\begin{proof}
For $\eps,\delt\in\left(0,\frac{1}{2}\right)$ let $\Lambda_{MV}$
be the algorithm of McGregor and Vorotnikova \cite{mcgregor_et_al:OASIcs:2018:8295},
amplified to have success probability $1-\delt$, providing $\left(\left(1+\eps\right)\left(\alpha+2\right),\delt\right)$-approximation
for maximum-matching size in graphs with arboricity at most $\alpha$.
As shown in Corollary \ref{cor:MM=000026VC2AS}, $m\left(\cdot\right)$
is $2$-almost-smooth. Therefore, using Theorem \ref{thm:Almost-Smooth}
with $c=1,d=2,C=\alpha+2$ and algorithm $\Lambda_{MV}$, we obtain
a sliding window algorithm $\Lambda$ which $\left(\left(2+\eps\right)\left(\alpha+2\right)^{2},\delt\right)$-approximate
the maximum-matching size in graphs with arboricity $\alpha$.

The space complexity of $\Lambda_{MV}$ is $s_{MV}\left(\eps,\delt\right)=\O{\eps^{-2}\log^{2}n\log\delt^{-1}}$
and it the update time is $t_{MV}\left(\eps,\delt\right)=\O{\eps^{-2}\log n\log\delt^{-1}}$.
Hence the space complexity of $\Lambda$ is
\[
\O{\eps^{-1}\log w\cdot s_{MV}\left(\eps,\frac{\eps\delt}{2w\log w}\right)}=\O{\eps^{-3}\log^{4}n\log\tfrac{1}{\eps\delt}},
\]
and similarly for the update time, where we used the fact that $w\le n^{2}$.
\end{proof}
For the purpose of approximating the maximum-matching size in graphs
with arboricity bounded by $\alpha$ Cormode et al. \cite{cormode_et_al:LIPIcs:2017:7849}
introduced the notion of $\alpha$-good edges. The algorithm of \cite{mcgregor_et_al:OASIcs:2018:8295}
used in the above proof actually approximates the maximum number of
$\alpha$-good edges in prefixes of the stream. Thus, using the same
algorithm of \cite{mcgregor_et_al:OASIcs:2018:8295}, we can directly
approximate the maximum size of the set of $\alpha$-good edges in
the active window $W$. For completeness we present here the definition
of Cormode et al. \cite{cormode_et_al:LIPIcs:2017:7849} for $\alpha$-good
edges in a stream, and the notion of $E_{\alpha}^{*}$ due to McGregor
and Vorotnikova \cite{mcgregor_et_al:OASIcs:2018:8295}.
\begin{defn}
Let $\S=\left(e_{1},e_{2},\ldots,e_{k}\right)$ be a sequence of $k$
edges on the set of vertices $V=\left[n\right]$. We say that an edge
$e_{i}=\left\{ u,v\right\} $ is $\alpha$-good (with respect to the
stream $\S$) if $d_{i}\left(u\right)\le\alpha$ and $d_{i}\left(v\right)\le\alpha$,
where $d_{i}\left(x\right)$ is the number of edges incident on the
vertex $x$ that appear after edge $e_{i}$ in the stream, i.e., $d_{i}\left(x\right)=\left|\left\{ e_{j}\mid j>i\land x\in e_{j}\right\} \right|$.
Denote by $E_{\alpha}\left(\S\right)$ the set of $\alpha$-good edges
in the stream $S$, and let $E_{\alpha}^{*}\left(\S\right)=\max\limits _{t\in\left[k\right]}\left|E_{\alpha}\left(\S_{t}\right)\right|$,
where $\S_{t}=\left(e_{1},e_{2},\ldots,e_{t}\right)$ is the prefix
of $S$ of length $t$.
\end{defn}

Although the size of the set of $\alpha$-good edges in a stream is
not smooth or even almost-smooth, the function $E_{\alpha}^{*}\left(\cdot\right)$
is almost-smooth, since it is subadditive.
\begin{lem}
\label{lem:E_almost_smoth}The function $E_{\alpha}^{*}\left(\cdot\right)$
is $2$-almost-smooth.
\end{lem}

\begin{proof}
Obviously $E_{\alpha}^{*}\left(\cdot\right)$ is non-negative and
bounded. It is also monotone, since it is defined by taking a maximum
of prefixes and earlier edges do not interfere with later edges being
$\alpha$-good. Hence, we are only left to show that it is indeed
subadditive. Let $A$ and $B$ be disjoint segments of the stream
$S=\left(e_{1},e_{2},\ldots,e_{k}\right)$. If $E_{\alpha}^{*}\left(AB\right)=E_{\alpha}^{*}\left(A\right)$
then obviously $E_{\alpha}^{*}\left(A\right)+E_{\alpha}^{*}\left(B\right)\ge E_{\alpha}^{*}\left(AB\right)$,
as $E_{\alpha}^{*}$ is non-negative. Otherwise, let $1\le t\le k$
be such that $e_{t}\in B$ and $E_{\alpha}^{*}\left(AB\right)=\left|E_{\alpha}\left(\left(AB\right)_{t}\right)\right|$,
then
\[
E_{\alpha}\left(\left(AB\right)_{t}\right)=\left(E_{\alpha}\left(\left(AB\right)_{t}\right)\cap A\right)\uni E_{\alpha}\left(B_{t}\right),
\]
where it is a disjoint union. Note that $E_{\alpha}\left(\left(AB\right)_{t}\right)\cap A\sseq E_{\alpha}\left(A\right)$,
as every $\alpha$-good edge from $A$ with respect to the stream
$\left(AB\right)_{t}$ is also $\alpha$-good edge in the stream $A$.
Hence $E_{\alpha}^{*}\left(\cdot\right)$ is subadditive,
\[
E_{\alpha}^{*}\left(AB\right)\le\left|E_{\alpha}\left(A\right)\right|+\left|E_{\alpha}\left(B_{t}\right)\right|\le E_{\alpha}^{*}\left(A\right)+E_{\alpha}^{*}\left(B\right).
\]
Therefore, using Lemma \ref{lem:subadditive2AS} we deduce that it
is indeed $2$-almost-smooth, as required.
\end{proof}
McGregor and Vorotnikova \cite{mcgregor_et_al:OASIcs:2018:8295} proved
that $m\left(\S\right)\le\left|E_{\alpha}\left(\S\right)\right|\le\left(\alpha+2\right)\cdot m\left(\S\right)$
for every stream $\S$, and thus also $m\left(\S\right)\le E_{\alpha}^{*}\left(\S\right)\le\left(\alpha+2\right)\cdot m\left(\S\right)$.
They also designed for $E_{\alpha}^{*}\left(\cdot\right)$ a $\left(1+\eps,\delt\right)$-approximation
algorithm with space complexity $s_{MV}\left(\eps,\delt\right)$,
as explained in Theorem \ref{thm:a^2-approx}. Since $E_{\alpha}^{*}\left(\cdot\right)$
is $2$-almost-smooth by Lemma \ref{lem:E_almost_smoth} we can apply
Theorem \ref{thm:Ext-Almost-Smooth}, with $g=E_{\alpha}^{*}\left(\cdot\right)$
and $f=m\left(\cdot\right)$, to obtain the following improvement
over Theorem \ref{thm:a^2-approx}.
\begin{thm}
[Restatement of Theorem \ref{thm:MM_SW}] \label{thm:improved_thm}
For every $\eps,\delt\in\left(0,\frac{1}{2}\right)$, there is a sliding-window
$\left(\left(2+\eps\right)\left(\alpha+2\right),\delt\right)$-estimation
algorithm for the maximum-matching size in a graph with arboricity
$\alpha$, with space bound $\O{\eps^{-3}\log^{4}n\log\tfrac{1}{\eps\delt}}$
and update time $\O{\eps^{-3}\log^{3}n\log\tfrac{1}{\eps\delt}}$.
\end{thm}

\begin{rem}
\label{rem:2_app_of_E^*}For arboricity $\alpha=1$ we can achieve
better approximation ratio. Cormode et al. \cite{cormode_et_al:LIPIcs:2017:7849}
showed that in this case $m\left(\S\right)\le\left|E_{1}\left(\S\right)\right|\le2\cdot m\left(\S\right)$
and thus $m\left(\S\right)\le E_{1}^{*}\left(\S\right)\le2\cdot m\left(\S\right)$.
Therefore, by Theorem \ref{thm:Ext-Almost-Smooth} there is a $\left(4+\eps,\delt\right)$-approximation
algorithm for the maximum-matching size in forest graphs in the sliding-window
model with the same space and update time bounds.
\end{rem}

\section{Applications to Minimum Vertex-Cover}

We show here few results for minimum vertex-cover (again in the sliding-window
model), based on its relationship to maximum and maximal matching,
and the fact that it is also almost smooth, see Corollary \ref{cor:MM=000026VC2AS}.
We start by showing an algorithm with approximation factor $3.125+\eps$
for the size of a minimum vertex-cover in VDP graphs using $\Ot{\sqrt{n}}$
space. We continue and present another algorithm for a larger family
of graphs, namely, forest graphs, where the approximation factor grows
to $4+\eps$ but the space complexity reduces to $\polylog\left(n\right)$.
We then proceed to show how to report a feasible vertex cover. We
reproduce a known algorithm for general graphs with approximation
factor $8+\eps$ that computes a vertex cover using $\Ot n$ space.
Then we show how to improve the approximation factor to $4+\eps$
by a tighter analysis of that same algorithm, using that the size
of a greedy maximal matching is also almost smooth.

There are two different but related problems to consider. The first
one is estimating the size of the minimum vertex-cover (without providing
a corresponding vertex cover of that size), and the second one is
computing a feasible vertex-cover of approximately minimum size.

Recall that the minimum vertex-cover size is almost-smooth, since
it is subadditive, as shown in Corollary \ref{cor:MM=000026VC2AS}.
Hence, we can use the Almost-Smooth-Histogram approach to estimate
the size of the minimum vertex-cover in the sliding-window model,
as explain in the next section.

\subsection{Vertex-Cover Estimation}

First we consider estimating the size of the minimum vertex cover
in the sliding-window model. We provide the first sub-linear space
algorithm in the sliding-window model for estimating $VC\left(\cdot\right)$,
for some families of graphs, as explained below.

A graph $\G$ is said to be VDP (stands for vertex-disjoint paths)
if $G$ is a union of vertex disjoint paths. We show two sliding-window
algorithms for different families of graphs. One with $\Ot{\sqrt{n}}$
space obtaining almost $3.125$-approximation for the family of VDP
graphs and the other one with $\polylog\left(n\right)$ space obtaining
almost $4$-approximation for graphs of arboricity $\alpha=1$. Observe
that the results are incomparable, since the first algorithm has better
approximation ratio but its space complexity is much bigger. Also,
the second algorithm is applicable for a larger family of graphs.

For the family of VDP graphs there is a randomized algorithm in the
turnstile streaming model to approximate $VC\left(\cdot\right)$,
presented in \cite{Otniel16}. Using the standard argument of computing
a median of $\log\delt^{-1}$ parallel repetitions, to achieve low
failure probability $\delt$, we can state this result as follows.
\begin{thm}
\cite[Theorem 1.1]{Otniel16} For every $\eps\in\left(0,\frac{1}{2}\right),\delt\in\left(0,1\right)$,
there exists a turnstile $\left(\frac{5}{4}+\eps,\delt\right)$-approximation
streaming algorithm for $VC\left(\cdot\right)$ in VDP graphs with
space bound $\O{\eps^{-1}\sqrt{n}\log^{2}n\log\delt^{-1}}$.
\end{thm}

Using Theorem \ref{thm:Almost-Smooth}, with $c=1,d=2$ and $C=\frac{5}{4}$
(thus $dc^{2}C^{2}=3.125$), we obtain as a corollary the following
result for the sliding-window model.
\begin{thm}
\label{thm:SW-forVC-sqrt} For every $\eps\in\left(0,\frac{1}{2}\right),\delt\in\left(0,1\right)$,
there exists a sliding-window $\left(3.125+\eps,\delt\right)$-approximation
algorithm for $VC\left(\cdot\right)$ in VDP graphs with space bound
$\O{\eps^{-2}\sqrt{n}\log^{4}n\log\frac{1}{\eps\delt}}$.
\end{thm}

Observe that a VDP graph has arboricity $\alpha=1$, because it is
a forest, and in particular it is a bipartite graph. Recall that according
to K\H{o}nig's theorem, in a bipartite graph the size of a minimum
vertex cover equals the size of a maximum-matching. Therefore, we
conclude from Remark \ref{rem:2_app_of_E^*} that there is a $\left(4+\eps,\delt\right)$-approximation
algorithm for the minimum vertex cover in VDP graphs using $\polylog$
space. Obviously it extends to all forests, i.e., graphs with arboricity
$\alpha=1$. Comparing to Theorem \ref{thm:SW-forVC-sqrt}, the following
theorem has slightly worse approximation factor but its space complexity
is much better, moreover, its applicable for a wider family of graphs.
\begin{thm}
\label{thm:SW-forVC-polylog} For every $\eps,\delt\in\left(0,\frac{1}{2}\right)$,
there is a sliding-window $\left(4+\eps,\delt\right)$-approximation
algorithm for the size of the minimum vertex-cover size in a forest
graph, with space bound $\O{\eps^{-3}\log^{4}n\log\tfrac{1}{\eps\delt}}$
and update time $\O{\eps^{-3}\log^{3}n\log\tfrac{1}{\eps\delt}}$.
\end{thm}

\begin{rem}
\label{rem:VCforForests}According to the previous paragraph and Remark
\ref{rem:2_app_of_E^*}, the algorithm of Cormode et al. \cite{cormode_et_al:LIPIcs:2017:7849}
$\left(2+\eps\right)$-approximates the minimum vertex-cover size
in forest graphs (arboricity $\alpha=1$) in the insertion-only model
with space $\O{\eps^{-2}\log^{2}n}$ and update time $\O{\eps^{-2}\log n}$.
\end{rem}

\subsection{Vertex-Cover Approximation}

Here we consider computing a feasible vertex cover of approximately
minimum size. We improve the approximation ratio of the algorithm
of \cite{Otniel16} from $8+\eps$ to $4+\eps$, using a tighter analysis
of his algorithm.

A $\emph{maximal matching}$ is a matching that cannot be extended
by adding an edge to it, i.e., a matching $M$ in a graph $\G$ is
maximal if every edge $e\in E\backslash M$ is adjacent to at least
one edge from the matching $M$. For a stream $A$ of edge insertions,
denote by $\widehat{M}\left(A\right)$ the greedy matching on $A$,
and denote by $\widehat{m}\left(A\right)$ its size. Notice that for
every stream $A$ the greedy matching $\widehat{M}\left(A\right)$
is maximal. Recall that for a matching $M$ we denote by $V\left(M\right)$
the set of all endpoints of edges from $M$, i.e., $V\left(M\right)=\left\{ v\in V\mid\exists u\in V,\left\{ v,u\right\} \in M\right\} $.

We first show that the greedy-matching size of a stream of edge insertions
is almost-smooth. The proof is similar in nature to that of Corollary
\ref{cor:MM=000026VC2AS}, but different because $\widehat{m}\left(\cdot\right)$
is not left-monotone, but rather $2$-left-monotone. Furthermore,
we can use the actual matching, since it is well structured.
\begin{lem}
\label{lem:greedy-matching_(2,2)-as}The greedy-matching size is $\left(2,2\right)$-almost-smooth.
\end{lem}

\begin{proof}
The first and third requirements are clear, as $\widehat{m}\left(\cdot\right)$
is non-negative and bounded. For the $2$-monotonicity, let $A,B$
be disjoint segments of the stream. Note that for every $e\in\widehat{M}\left(B\right)$
at least one of its endpoints is in $V\left(\widehat{M}\left(AB\right)\right)$,
hence $\widehat{m}\left(B\right)\le\left|V\left(\widehat{M}\left(AB\right)\right)\right|=2\cdot\widehat{m}\left(AB\right)$.

To prove the almost-smoothness property, let $A,B$ and $C$ be disjoint
segments of the stream, and suppose $\eps\widehat{m}\left(AB\right)\le\widehat{m}\left(B\right)$
for some $\eps\in\left(0,1\right)$. We first claim that
\[
\widehat{m}\left(ABC\right)\le\widehat{m}\left(AB\right)+\widehat{m}\left(BC\right).
\]
The lefthand-side counts the edges in $\widehat{M}\left(ABC\right)$,
which can be partitioned into edges from $AB$ and from $C$. The
former set is, by construction, $\widehat{M}\left(ABC\right)\cap AB=\widehat{M}\left(AB\right)$,
hence we actually need to show that
\[
\left|\widehat{M}\left(ABC\right)\cap C\right|\le\left|\widehat{M}\left(BC\right)\right|.
\]
To prove this, we map each edge $e\in\widehat{M}\left(ABC\right)\cap C$
to an edge $f\left(e\right)\in\widehat{M}\left(BC\right)$, defined
as the first (earliest) edge in the stream $BC$ that intersects $e$
(i.e., shares at least one endpoint with $e$, possibly $e$ itself).
Observe that $f\left(e\right)$ is well-defined because the greedy
matching $\widehat{M}\left(BC\right)$ considers at some point the
edge $e\in C$ itself; moreover, if $e\in\widehat{M}\left(BC\right)$
then $f\left(e\right)=e$, and otherwise $f\left(e\right)\neq e$.
Assume for contradiction that $f\left(e_{1}\right)=f\left(e_{2}\right)$
for two distinct edges $e_{1}\neq e_{2}\in\widehat{M}\left(ABC\right)\cap C$,
and without loss of generality, let $e_{1}$ appear before $e_{2}$
in the stream $C$. Observe that $f\left(e_{1}\right)\notin\widehat{M}\left(ABC\right)$
because it intersects two distinct edges $e_{1},e_{2}$ in that same
matching $\widehat{M}\left(ABC\right)$ (having two distinct edges
handles the possibility that $f\left(e_{1}\right)$ is one of $e_{1},e_{2}$).
This means that there is some $e_{0}\in\widehat{M}\left(ABC\right)$
that intersects $f\left(e_{1}\right)$ and appears before it in the
stream $ABC$. Notice that this $e_{0}\neq e_{1},e_{2}$ because $e_{0}$
appears before $f\left(e_{1}\right)$ (in $ABC$) which in turn appears
no later than $e_{1}$ and $e_{2}$ (in $BC$). We thus have three
distinct edges $e_{0},e_{1},e_{2}\in\widehat{M}\left(ABC\right)$,
all intersecting $f\left(e_{1}\right)$, which has only two endpoints,
therefore at least two of those three edges must intersect each other,
in contradiction to $\widehat{M}\left(ABC\right)$ being a matching.
This implies that $f$ is injective, and consequently the claimed
inequality.

Let us now complete the proof of the lemma. By construction, $\widehat{m}\left(B\right)\le\widehat{m}\left(BC\right)$,
and we conclude that
\[
\begin{alignedat}{1}2\widehat{m}\left(BC\right)\ge & \widehat{m}\left(B\right)+\widehat{m}\left(BC\right)\ge\eps\widehat{m}\left(AB\right)+\widehat{m}\left(BC\right)\\
\ge & \eps\left(\widehat{m}\left(AB\right)+\widehat{m}\left(BC\right)\right)\ge\eps\widehat{m}\left(ABC\right).
\end{alignedat}
\]
\end{proof}
We proceed to show an approximation algorithm for minimum vertex-cover
using our almost-smoothness framework, as stated in the next theorem.
Its proof basically transforms an approximation algorithm for maximum-matching
to one for minimum vertex-cover. The proof starts with a general technique,
and then refines the analysis by exploiting the specifics of our setting.
\begin{thm}
\label{thm:SW_VC} For every $\eps\in\left(0,\frac{1}{2}\right)$,
there is a sliding-window $\left(4+\eps\right)$-approximation algorithm
for the minimum vertex cover with space bound $\O{\eps^{-1}n\log^{2}n}$.
\end{thm}

\begin{rem}
\label{rem:feasible_solution_algorithm}Below, we wish to use the
algorithm presented in Theorem \ref{thm:Almost-Smooth}, however that
algorithm returns a number, while we need an algorithm that returns
a feasible solution. This disparity is easy to resolve for greedy-matching,
by simply outputting the underlying vertices of the greedy-matching
associated with bucket $B_{1}$, instead of outputting only its size.
\end{rem}

\begin{proof}
[Proof of Theorem \ref{thm:SW_VC}]If $M^{*}\sseq E$ is a maximal
matching in the graph $\G$ then the set of vertices $V\left(M^{*}\right)$
is a vertex cover of the graph $G$, because every edge from $E$
has at least one of its endpoints in $V\left(M^{*}\right)$ (otherwise
the matching $M^{*}$ would not be maximal). For every stream $A$
the greedy matching $\widehat{M}\left(A\right)$ is a maximal matching
and thus $V\left(\widehat{M}\left(A\right)\right)$ is a vertex cover
of the edges from $A$. Hence, we refer to the greedy-matching algorithm
also as the greedy vertex cover algorithm, with the only difference
that it outputs the vertices $V\left(\widehat{M}\left(A\right)\right)$
of the matching, instead of the edges $\widehat{M}\left(A\right)$
of the matching.

The greedy vertex cover algorithm achieves $2$-approximation in the
standard insertion-only streaming model for the minimum vertex cover
using $\O{n\log n}$ space, because at least one vertex from each
matched edge must be in the minimum vertex cover. By using this greedy
algorithm and the $2$-almost-smoothness of the minimum vertex cover
size, we can deduce from the variant of Theorem \ref{thm:Almost-Smooth}
discussed in Remark \ref{rem:feasible_solution_algorithm}, an $\left(8+\eps\right)$-approximation
algorithm for reporting a minimum vertex cover in the sliding-window
model with $\O{\eps^{-1}n\log^{2}n}$ space, matching the result of
\cite{Otniel16}. 

We can improve the approximation ratio by using the algorithm of Crouch
et al. \cite{Crouch2013}, which achieves $\left(3+\eps\right)$-approximation
to maximum-matching, using the same space complexity. Their algorithm
maintains a greedy matching in various buckets, such that the difference
between adjacent buckets is not too large. Specifically, for any adjacent
buckets $B_{i}$ and $B_{i+1}$ it holds that $2\widehat{m}\left(B_{i+1}\right)\ge\left(1-\eps\right)\widehat{m}\left(B_{i}\right)$.
By an easy modification to their algorithm, just outputting the greedy
matching on the bucket $B_{1}$ instead of the bucket $B_{2}$, it
holds that $V\left(\widehat{M}\left(B_{1}\right)\right)$ is a vertex
cover (of $B_{1}\supseteq W$) at most $\left(6+\eps\right)$-factor
larger than the minimum vertex cover on the active window $W$. Note
that the algorithm of \cite{Crouch2013} and the algorithm of \cite{Otniel16}
are essentially the same, the only difference is that \cite{Otniel16}
stores the vertices instead of the edges.

We can improve even further, to $\left(4+\eps\right)$-approximation,
by leveraging the fact that the greedy-matching size is $\left(2,2\right)$-almost-smooth.
Indeed, let us use the algorithm of Crouch et al. \cite{Crouch2013},
but output $V\left(\widehat{M}\left(B_{1}\right)\right)$ instead
of $B_{2}$. At the end of the stream we are guaranteed that $2\widehat{m}\left(B_{2}\right)\ge\left(1-\eps\right)\widehat{m}\left(B_{1}\right)$,
because the greedy matching size is $\left(2,2\right)$-almost-smooth.
Since the minimum vertex cover is monotone and $W\sseq B_{1}$ $VC\left(W\right)\le VC\left(B_{1}\right)\le\left|V\left(\widehat{M}\left(B_{1}\right)\right)\right|=2\cdot\widehat{m}\left(B_{1}\right)$.
Note that $V\left(\widehat{M}\left(B_{1}\right)\right)$ is indeed
a vertex cover on the active window $W$, since it is a vertex cover
on $B_{1}$. Additionally, $VC\left(W\right)\ge VC\left(B_{2}\right)\ge\widehat{m}\left(B_{2}\right)$,
since $VC\left(\cdot\right)$ is monotone, $B_{2}\sseq W$ and the
minimum vertex cover size is at least the size of any matching. For
$\eps<\frac{1}{2}$ it holds that $\frac{1}{1-\eps}\le1+2\eps$, and
we obtain
\[
\left|V\left(\widehat{M}\left(B_{1}\right)\right)\right|=2\cdot\widehat{m}\left(B_{1}\right)\le4\left(1+2\eps\right)\cdot\widehat{m}\left(B_{2}\right)\le4\left(1+2\eps\right)\cdot VC\left(W\right).
\]
We conclude that the output $V\left(\widehat{M}\left(B_{1}\right)\right)$
is a vertex cover on the active window $W$ and it is at most a factor
$4\left(1+2\eps\right)$ larger then $VC\left(W\right)$.
\end{proof}

\paragraph{Acknowledgments}

We thank Oded Goldreich and Shahar Dobzinski for suggesting to generalize
Corollary \ref{cor:MM=000026VC2AS} to subadditive functions as presented
in Lemma \ref{lem:subadditive2AS}. We also thank Sepehr Assadi and
Krishna Chaitanya for pointing out an error in our earlier proof of
Lemma \ref{lem:greedy-matching_(2,2)-as}.

\bibliographystyle{bib/alphaurlinit}
\addcontentsline{toc}{section}{\refname}\bibliography{bib/mybib}

\newpage{}

\appendix

\section{\label{appendix:Smooth-Histogram-Framework}Smooth-Histogram Framework}

The smooth-histogram technique presented by Braverman and Ostrovsky
\cite{Braverman:2007:SHS:1333875.1334202} is one of only two general
techniques for adapting insertion-only algorithms to the sliding-window
model. The other one is an earlier technique called exponential histogram,
due to Datar et al. \cite{StreamStatistics}. The approach of \cite{Braverman:2007:SHS:1333875.1334202}
is to maintain several instances of an insertion-only algorithm on
different suffixes of the stream, such that at every point in time,
the algorithm can output an approximation of $f$ on $W$. They showed
that for a large family of functions, which they called smooth (see
Remark \ref{rem:(a,b)-smooth}), this approach yields a good approximation
algorithm for the sliding-window model. Their technique yields no
results for functions that are not smooth, specifically graph problems
such as maximum-matching size. We extend this framework to a much
broader family of functions that we call almost-smooth.

More precisely, assume there is an algorithm $\Lambda$ that $C$-approximates
a left-monotone non-decreasing function $f$ in the insertion-only
model. The smooth-histogram framework (for the sliding-window model)
maintains $k=\O{\eps^{-1}\log w}$ instances of $\Lambda$. Each instance
$\Lambda_{i}$ processes the stream from some initial point in time
until the end of the stream (or until it is discarded), i.e., it corresponds
to some suffix of the stream, referred to as a $\emph{bucket}$. The
bucket corresponding to $\Lambda_{i}$ is denoted by $B_{i}$, and
we denote by $\Lambda_{i}\left(B_{i}\right)$ the value of instance
$\Lambda_{i}$ on the stream $B_{i}$. These buckets will satisfy
the invariant $B_{1}\supseteq W\supsetneq B_{2}\supsetneq B_{3}\supsetneq\cdots\supsetneq B_{k}$,
where $W$ is the active window. In order to use only a small amount
of space, whenever two nonadjacent instances have ``close'' values,
all instances between them will be deleted. Instances $\Lambda_{i}$
and $\Lambda_{j}$, for $j>i$, are considered close if $\Lambda_{i}\left(B_{i}\right)$
and $\Lambda_{j}\left(B_{j}\right)$ are within factor $1+\eps$ of
each other. At each step of receiving a new item from the stream,
the sliding-window algorithm updates all the instances, creates a
new instance $\Lambda_{k+1}$, which initially contains only the new
item, deletes all unnecessary instances, as explained above, and lastly
renumbers the buckets (consecutively starting from $1$). For a more
elaborate description see Algorithm \ref{alg:Update-and-query} in
Section \ref{sec:SW-alg}.

We show that applying this approach to almost-smooth functions yields
good approximation algorithms while only storing a small number of
buckets. Intuitively, $\Lambda_{1}\left(B_{1}\right)$ approximates
$f\left(W\right)$ (up to some factor that depends on $d,C$ and $\eps$)
because $\Lambda_{1}\left(B_{1}\right)$ and $\Lambda_{2}\left(B_{2}\right)$
are close up to some factor (since deleted buckets have close value
to nearby buckets by the almost-smoothness of $f$) and thus they
bound $\Lambda\left(W\right)$. Therefore, by deleting buckets between
close instances we ensure that the number of buckets is small while
the approximation ratio is roughly $dC^{2}$.

Braverman and Ostrovsky \cite{Braverman:2007:SHS:1333875.1334202}
proved that all $\ell_{p}$-norms, for $p>0$, are smooth (in our
terminology it means almost-smoothness parameter $d=1$) and consequently
obtained algorithms that $\left(1+\eps\right)$-approximate these
norms in the sliding-window model, with an overhead (relative to insertion-only
algorithms) of roughly factor $\O{\eps^{-1}\log w}$ in the space
complexity.

While they analyze their framework only for smooth functions (such
as $\ell_{p}$-norms) our analysis considers the larger family of
$d$-almost-smooth function (which includes all the subadditive functions).
Many graph problems are $2$-almost-smooth (as they are subadditive)
but not smooth, and thus do not fit their analysis. Additionally,
they do not consider functions that have only a $C$-approximation
algorithms in the insertion-only model for constant $C>1$. We analyze
the dependence on $C$ and present here the first sliding-window algorithms
for such functions.

We point out that previously studied techniques for estimating ``weakly
superadditive'' functions, such as $\ell_{p}$ norms for $p\in\left[1,2\right]$
(e.g. the exponential-histogram technique of Datar et al. \cite[Section 6]{StreamStatistics}),
are not relevant to our study of subadditive functions. For example,
maximum-matching is subadditive but not weakly superadditive.

\section{\label{appendix:Proof-of-Main-Theorem}Proof of Theorem \ref{thm:Almost-Smooth}}

To avoid dependence on the length of the entire stream (for the success
probability) we make use of a general observation due to Braverman
regarding algorithms for the sliding-window model. Intuitively, it
says that without loss of generality, the entire stream can be assumed
to have length at most twice the size of the window.
\begin{claim}
\label{claim:SW_bound_by_2w} Every sliding-window algorithm $\Lambda$
can be modified such that it will not depend on the length of the
entire stream, but only depend on at most $2w$ last items from the
stream, while using at most a factor $2$ more space.
\end{claim}

\begin{proof}
[Proof of Claim \ref{claim:SW_bound_by_2w}] To avoid dependence on
the length of the stream $N$, and instead be dependent only on the
length of the window $w$, we can argue as follows: partition the
entire stream $D$ to segments $D_{1},D_{2},\ldots,D_{t}$, of length
$w$ each, where $t=\left\lceil \frac{N}{w}\right\rceil $ (except
maybe the last segment $D_{t}$, which is of length $0<N-\left(t-1\right)w\le w$).
At each segment $D_{i}$ start a new instance of algorithm $\Lambda$,
and keep running it during the next segment as well, for at most $2w$
updates in total (for each instance of $\Lambda$). At any point in
time, to answer a query the algorithm queries the instance of $\Lambda$
on the penultimate segment, which corresponds to a suffix of the stream
of length at least $w$, and thus contains the entire active window.
Thus, at each point in time it is enough to store only the two instances
of algorithm $\Lambda$ corresponding to the last two segments, increasing
the storage requirement only by a factor of $2$.
\end{proof}
\enlargethispage{4\baselineskip}
\begin{proof}
[Proof of Theorem \ref{thm:Almost-Smooth}] Assume, without loss of
generality, that the length of the entire stream is at most $2w$,
as explained in Claim \ref{claim:SW_bound_by_2w}. Denote by $\Lambda(X)$
the output of algorithm $\Lambda$ when run on the stream $X$. Assume
that $\Lambda$ has $\frac{\varepsilon\delt}{2w\log w}$ failure probability
and $\eps$ is the accuracy parameter, i.e., it $\left(\left(1+\eps\right)C,\frac{\eps\delt}{2w\log w}\right)$-approximates
$f$. Recall that we use the term ``bucket'' to refer to a suffix
of the stream. Our algorithm maintains (not explicitly) $k=O\left(\varepsilon^{-1}\log n\right)$
``buckets'' $B_{1},\ldots,B_{k}$. At all points in time, these
buckets will satisfy the invariant $B_{1}\supseteq W\supsetneq B_{2}\supsetneq B_{3}\supsetneq\cdots\supsetneq B_{k}$,
where $W$ is the active window. For each bucket $B_{i}$ the algorithm
maintains an instance of $\Lambda$, denoted by $\Lambda_{i}$. In
order to use only a small amount of space, whenever two nonadjacent
buckets have similar value according to $\Lambda$ we will delete
all buckets between them. For ease of exposition, the algorithm will
be defined using these buckets, and later we explain how to not actually
store the buckets themselves. In each step of receiving a new item
$a$ from the stream, the algorithm updates the current buckets $B_{1},\ldots,B_{k}$
and the corresponding instances $\Lambda_{1}(B_{1}),\ldots,\Lambda_{k}(B_{k})$
in the following way. \\

\begin{algorithm}[H]
\caption{Almost-Smooth Histogram for Sliding-Window Streaming} \label{alg:Update-and-query}
\DontPrintSemicolon

\def\NoNumber#1{{\def\alglinenumber##1{}\State #1}\addtocounter{ALG@line}{-1}}

\BlankLine
{\nonl \bf  Initialization Procedure:}

$k \letdef 0$ (no buckets exist yet) \;

\BlankLine \BlankLine

{\nonl \bf Update Procedure:}

open a new bucket $ B_{k+1} \letdef \left\{ a\right\} $, and start a new instance $ \Lambda_{k+1} $ on this bucket \; 

add $a$ to every bucket $B_{i}$, $i \in [k]$, and update instance $ \Lambda_{i} $ accordingly \;

$ i  \letdef  1 $ \;

\While {$ i \le k-2 $}{

	find the largest $j \ge i$ such that $\Lambda_{j}(B_{j})>\left(1-\varepsilon\right)\Lambda_{i}(B_{i})$ \;
	
foreach $ t = i+1, \ldots ,j-1$ discard bucket $B_{t}$ and its associated instance $\Lambda_{t}$ \;
	$ i  \letdef  \min \left\{ j,i+1\right\} $ \;
}

\If{$W \sseq B_{2}$}{
	discard bucket $B_{1}$ and its associated instance $\Lambda_{1}$ \;
}

let $k$ be the number of remaining buckets, and renumber the buckets and their associated instances (keeping their order) to $ B_{1}, \ldots ,B_{k} $ \;

\BlankLine \BlankLine

{\nonl \bf Query Procedure:}

\BlankLine

\If{$B_{1} = W$}{
	\Return $\Lambda_{1}(B_{1})$ \;
}

\Else{
	\Return $dcC\frac{\left(1+\eps\right)}{\left(1-\eps\right)^{2}}\cdot\Lambda_{2}(B_{2})$ \;
}

\end{algorithm}

\BlankLine

Since $f$ is bounded by some polynomial in $w$, and $\Lambda_{i+2}(B_{i+2})\le\left(1-\varepsilon\right)\Lambda_{i}(B_{i})$
for every $1\le i\le k-2$ (as the ``unnecessary'' instances were
deleted in the process of updating), it follows that the number of
instances are bounded by $\O{\eps^{-1}\log w}$. Hence, the number
of times any instance of $\Lambda$ is invoked is at most $\O{\tfrac{1}{\eps}\log w}\cdot m$,
where $m$ is the length of the entire stream, which we assumed to
be bounded by $2w$. Since we set the failure probability to be $\frac{\eps\delt}{2w\log w}$
then by union bound the probability that any invocation of any instance
of $\Lambda$ fails is $\delt$, i.e., with probability $1-\delt$
every instance of $\Lambda$ succeeds every time it is invoked. Thus,
from now on we assume that every instance of $\Lambda$ succeeds every
time, i.e., it $\left(1+\eps\right)C$-approximates $f$ on the corresponding
bucket whenever it is invoked.

Now, let us explain how to achieve $\O{\eps^{-1}\log w\cdot\left(s\left(\eps,\frac{\eps\delt}{2w\log w}\right)+\log w\right)}$
space complexity. We can directly update the instances $\Lambda_{1},\ldots,\Lambda_{k}$,
without storing the buckets $B_{1},\ldots,B_{k}$ explicitly, hence
we only need to maintain the storage that algorithm $\Lambda$ requires.
Additional to the space required by instances $\Lambda_{1},\ldots,\Lambda_{k}$,
we need to store a counter $c_{i}$ for every instance $\Lambda_{i}$,
indicating its initialization time (initialized to $c_{i}=1$ and
incremented each time a new item arrives), such that we can perform
the last step of the algorithm, by comparing the counter of bucket
$B_{2}$ to the number $w$ (which is the size of the active window
$W$). This way, for each bucket $B_{i}$ we store $s\left(\eps,\frac{\eps\delt}{2w\log w}\right)+\log w$
bits. As we have seen previously, the number of instances of $\Lambda$
are bounded by $\O{\eps^{-1}\log w}$. Therefore, the total number
of bits used by the algorithm is $\O{\eps^{-1}\log w\cdot\left(s\left(\eps,\frac{\eps\delt}{2w\log w}\right)+\log w\right)}$,
as claimed.

For the approximation ratio denote by $\tilde{\Lambda}$ the output
of the algorithm, and note that either $\tilde{\Lambda}=\Lambda_{1}(B_{1})$
in the case $B_{1}=W$, or $\tilde{\Lambda}=dcC\frac{\left(1+\eps\right)}{\left(1-\eps\right)^{2}}\cdot\Lambda_{2}(B_{2})$
otherwise. If $B_{1}=W$ then $\tilde{\Lambda}=\Lambda(B_{1})$ is
obviously $\left(1+\eps\right)C$-approximates $f$ on $B_{1}=W$.
Otherwise $B_{1}\supsetneq W\supsetneq B_{2}$, which means that at
some earlier point in time, denoted by $t^{*}$, the algorithm had
deleted some buckets between buckets $B_{1}$ and $B_{2}$ to make
them adjacent (for the first time). For $i\in\left\{ 1,2\right\} $
denote by $B_{i}^{\prime}$ the bucket $B_{i}$ at the time $t^{*}$.
Let $D$ be the suffix of the stream starting at time $t^{*}$, and
observe that $B_{1}=B_{1}^{\prime}D$ and $B_{2}=B_{2}^{\prime}D$.
At time $t^{*}$ we had $\left(1-\eps\right)\Lambda(B_{1}^{\prime})<\Lambda(B_{2}^{\prime})$,
which implies
\[
\left(1-\eps\right)f(B_{1}^{\prime})\le\left(1-\eps\right)\Lambda(B_{1}^{\prime})<\Lambda(B_{2}^{\prime})\le\left(1+\eps\right)C\cdot f(B_{2}^{\prime}),
\]
namely $\frac{\left(1-\eps\right)}{\left(1+\eps\right)C}f(B_{1}^{\prime})\le f(B_{2}^{\prime})$.
Note that we used here the formulation of Remark \ref{rem:rand_alg}.
Since $f$ is $\left(c,d\right)$-almost-smooth, at the end of the
stream we have $\frac{1}{d}\cdot\frac{\left(1-\eps\right)}{\left(1+\eps\right)C}f(B_{1})\le f(B_{2})$.
Now, by monotonicity $\frac{1}{c}\cdot f(B_{2})\le f\left(W\right)\le c\cdot f(B_{1})$,
and altogether
\[
\frac{1}{cC\left(1+\eps\right)}\Lambda(B_{2})\le\frac{1}{c}\cdot f(B_{2})\le f(W)\le c\cdot f(B_{1})\le cdC\cdot\frac{\left(1+\eps\right)}{\left(1-\eps\right)}f(B_{2})\le cdC\cdot\frac{\left(1+\eps\right)}{\left(1-\eps\right)^{2}}\Lambda(B_{2}).
\]
Since $\frac{\left(1+\eps\right)^{2}}{\left(1-\eps\right)^{2}}\le1+20\eps$
for $\eps\le\frac{1}{2}$, we conclude that at the end of the stream
the output of the algorithm $\tilde{\Lambda}=dcC\frac{\left(1+\eps\right)}{\left(1-\eps\right)^{2}}\cdot\Lambda(B_{2})$
approximates $f(W)$ as claimed.
\end{proof}

\end{document}